\documentclass[
	runningheads,
	envcountsame
]{llncs}

\usepackage[T1]{fontenc}
\usepackage{newtxtext}
\usepackage[varvw]{newtxmath}
\usepackage{alphabeta}
\usepackage{csquotes}

\DeclareUnicodeCharacter{1EC5}{\~{ê}}
\let\epsilon\varepsilon

\usepackage{environ}
\usepackage{graphicx}
\usepackage{mathtools}
\usepackage{stmaryrd}
\usepackage{todonotes}
\usepackage{xspace}
\usepackage{xcolor}

\usepackage{ebproof}
	\NewDocumentCommand \drule {m} {\ensuremath{\mathsf{(#1)}}}
	\ebproofset{
	center=false,
	right label template=\small\drule{\inserttext},
	}
\NewEnviron{smallprooftree}[1][]{%
	\scalebox{.75}{\begin{prooftree}[#1]\BODY\end{prooftree}}%
}

\usepackage[inline]{enumitem}
\NewDocumentEnvironment{ienumerate}{}
	{\begin{enumerate*}[label=
		{\itshape (\roman*)} ]}
	{\end{enumerate*}}

\spnewtheorem{fact}[theorem]{Fact}{\bfseries}{\itshape}

\usepackage{thm-restate}
\makeatletter
\newcommand{\restateheading}[1]{#1}
\patchcmd{\@spbegintheorem}{#1\ #2}{\restateheading{#1\ #2}}{}{}
\patchcmd{\@spopargbegintheorem}{#1\ #2}{\restateheading{#1\ #2}}{}{}
\newcommand{\linkedrestate}[2]{%
  \begingroup
  \renewcommand{\restateheading}[1]{\hyperref[#1]{##1}}%
  #2*%
  \endgroup
}
\makeatother
\makeatletter
\NewDocumentCommand \thmmarginpar {m} {
	\ifthmt@thisistheone
	\leavevmode\marginpar{\scriptsize #1}%
	\fi
}
\makeatother

\makeatletter
\AddToHook{env/proof/begin}{%
  \def\@Opargbegintheorem#1#2#3#4{%
    \@Begintheorem{#2}{#3}{#4}%
  }%
}
\makeatother

\usepackage[
	backend=biber,
	style=lncs,
	citestyle=alphabetic,
]{biblatex}

\DeclareFieldFormat{labelalphawidth}{\mkbibbrackets{#1}}

\makeatletter
\patchcmd{\blx@env@bibliography}
  {\printfield[labelnumberwidth]{labelnumber}}
  {\printtext[labelalphawidth]{%
    \printfield{labelprefix}%
    \printfield{labelalpha}%
    \printfield{extraalpha}}}
  {}{\PackageError{biblatex}{Could not patch bibliography label}{}}
\patchcmd{\blx@env@bibliography}
  {\setlength{\labelwidth}{\labelnumberwidth}}
  {\setlength{\labelwidth}{\labelalphawidth}}
  {}{\PackageError{biblatex}{Could not patch bibliography width}{}}
\makeatother

\usepackage{hyperref}
\hypersetup{
	colorlinks, allcolors=blue
}
\usepackage{bookmark}
\hypersetup{bookmarksdepth=2}
\makeatletter
\def\toclevel@title{0}
\def\toclevel@section{1}
\def\toclevel@subsection{2}
\def\toclevel@subsubsection{3}
\makeatother
\AddToHook{cmd/appendix/after}{%
}

\usepackage[capitalise]{cleveref}
\crefname{subsection}{§}{§}

\NewDocumentCommand \eg				{}		{\emph{e.g.}\@\xspace}
\NewDocumentCommand \ie				{}		{\emph{i.e.}\@\xspace}
\NewDocumentCommand \ifandonlyif	{}		{iff\xspace}
\NewDocumentCommand \resp			{}		{resp.\@\xspace}
\NewDocumentCommand \wrt			{}		{wrt.\@\xspace}

\newlength{\widerelwidth}
\NewDocumentCommand \widerel		{om}	{ \IfValueTF{#1}
	{ \settowidth{\widerelwidth}{\ensuremath{#1}}
		\mathrel{\makebox[\widerelwidth][c]{\ensuremath{#2}}} }
	{ \quad #2 \quad } }

\NewDocumentCommand \eqdef			{}		{ \coloneqq }

\NewDocumentCommand \bnfni			{}		{ \widerel\ni }
\NewDocumentCommand \bnfeq			{}		{ \widerel\eqdef }
\NewDocumentCommand \bnfsep			{}		{ \mathrel\vert }

\NewDocumentCommand \enc			{m}		{ \ulcorner #1 \urcorner }

\NewDocumentCommand \set 			{smo}	{
	\IfBooleanT{#1}{\left} \{ #2
	\IfValueT{#3}{ \IfBooleanT{#1}{\middle} \mid #3 }
	\IfBooleanT{#1}{\right} \} }

\RenewDocumentCommand \setminus		{}		{ - }

\NewDocumentCommand \support		{}		{ \mathrm{supp} }

\NewDocumentCommand \Semiring		{}		{ \mathbf{S} }

\NewDocumentCommand \Nat			{s}
	{ \mathbf{N} \IfBooleanT{#1}{^{\infty}} }
\NewDocumentCommand \Reals			{st+}	
	{ \mathbf{R} \IfBooleanT{#1}{^{\infty}} \IfBooleanT{#2}{_{+}} }
\NewDocumentCommand \Tropical		{s}		
	{ \mathbf{T} \IfBooleanT{#1}{^*} }

\NewDocumentCommand \msets			{O{\Nat}m}	{ \mathcal M_{#1}(#2) }

\NewDocumentCommand \mset 			{smo}	{
	\IfBooleanT{#1}{\left} [ #2
	\IfValueT{#3}{ \IfBooleanT{#1}{\middle} \mid #3 }
	\IfBooleanT{#1}{\right} ] }

\NewDocumentCommand \ms				{m}		{ \mathsf{#1} }
\NewDocumentCommand	\gr				{mm}	{ #1:#2 }

\NewDocumentCommand \mscur			{m}		{ \ms{\tilde{#1}} }

\NewDocumentCommand \mstplus		{}		{ \oplus }
\NewDocumentCommand \mstsum			{}		{ \bigoplus }
\NewDocumentCommand \msttimes		{}		{ \odot }

\NewDocumentCommand \Var			{}		{ \mathcal{V} }
\NewDocumentCommand \abs			{mo}
	{ \mathnormal{\lambda} #1 \IfValueT{#2}{^{#2}} }

\NewDocumentCommand \pureterms		{}		{ \mathbf{\Lambda} }

\NewDocumentCommand \typeabs		{}		{ \mathnormal{\Lambda} }
\DeclareMathOperator \fold					{fold}
\DeclareMathOperator \unfold				{unfold}
\NewDocumentCommand \fd				{m}		{ \underline{#1} }
\NewDocumentCommand \ufd			{m}		{ \overline{#1} }

\NewDocumentCommand \fmuterms 		{}	
	{ \mathbf{\Lambda}_{\fatype\rectype} }

\NewDocumentCommand \termstyle		{m}		{ \mathtt{#1} }
\NewDocumentCommand \trueterm		{}		{ \termstyle{t} }
\NewDocumentCommand \falseterm		{}		{ \termstyle{f} }

\NewDocumentCommand \Atoms			{}		{ \mathcal{A} }
\NewDocumentCommand	\lto			{}		{ \multimap }

\NewDocumentCommand \lm				{}		{ \triangleright }

\NewDocumentCommand	\bang			{O{a}}	{ \mathord{!}_{#1} }
\NewDocumentCommand \ito			{omm}	{ \bang[#1] #2 \lto #3 }
\NewDocumentCommand \fatype			{}		{ \forall }

\NewDocumentCommand \rectype		{}		{ \mu }
\NewDocumentCommand \bangindex		{m}		{ g_{#1} }

\NewDocumentCommand \fmutypes 		{O{\Tropical}}
	{ \mathsf{G}_{#1} }

\NewDocumentCommand \typefont		{m}		{ \mathrm{#1} }
\NewDocumentCommand \booltype		{}		{ \typefont{Bool} }
\NewDocumentCommand \nattype		{s}
	{ \typefont{Nat}^{\IfBooleanT{#1}{*}} }
\NewDocumentCommand \strtype		{sO{}}
	{ \typefont{Str}_{#2}\IfBooleanT{#1}{'} }

\NewDocumentCommand \itypes {t!O{\Tropical}t+t-O{}}
	{%
	\IfBooleanT{#1}{\mathord{!}}
	\mathsf{I}_{\mathnormal{#2}}
	\IfBooleanT{#3}{^{+#5}}
	\IfBooleanT{#4}{^{-#5}}
	}

\NewDocumentCommand \inlater		{m}		{ \in_{#1} }

\NewDocumentCommand \derivstyle		{m}		{ \mathcal{#1} }
\NewDocumentCommand \derivD			{}		{ \derivstyle{D} }
\NewDocumentCommand \derivE			{}		{ \derivstyle{E} }

\NewDocumentCommand \derives		{t+O{}}
	{ \blacktriangleright^{\IfBooleanT{#1}{+#2}} }

\NewDocumentCommand \subjexp		{m}		{ \mathrm{SE}_{#1} }

\NewDocumentCommand \fv 			{}		{ \mathrm{fv} }

\NewDocumentCommand \subst 			{O{x}m}	{ [#2/#1] }
\NewDocumentCommand \tsubst 		{O{\alpha}m}	{ [#2/#1] }

\DeclarePairedDelimiter \forgetFmu {|} {|}

\NewDocumentCommand \dist			{}		{ \mathbf{d} }
\NewDocumentCommand \pairings		{mm}	{ \Pi(#1,#2) }
\NewDocumentCommand \deletions		{m}		{ \mathrm{del}(#1) }
\NewDocumentCommand \insertions		{m}		{ \mathrm{ins}(#1) }
\NewDocumentCommand \totpairing		{m}		{ \overline{#1} }
\NewDocumentCommand \cost			{}		{ \mathbf{c} }

\NewDocumentCommand \trunc 			{sO{d}m}
	{ \IfBooleanT{#1}{\left}\lfloor #3 \IfBooleanT{#1}{\right}\rfloor_{#2} }

\NewDocumentCommand \gradingpaths	{m}		{ \mathrm{GP}^{#1} }
\NewDocumentCommand \temporaldepth	{m}		{ \delta(#1) }
\NewDocumentCommand \mingrade		{}		{ \mathrm{ming} }
\NewDocumentCommand \msofpath		{}		{ \mathsf{target} }

\NewDocumentCommand \bred			{s}
	{ \rightarrow\IfBooleanT{#1}{^*} }

\NewDocumentCommand \hred			{sO{}}
	{ \rightarrow_{\mathrm{h}#2}\IfBooleanT{#1}{^*} }

\NewDocumentCommand \HNF			{}		{ \mathcal{HNF} }

\NewDocumentCommand \HN				{}		{ \mathcal{HN} }

\NewDocumentCommand \HHN			{}		{ \mathcal{HN}^+ }

\NewDocumentCommand \rc				{}		{\textsc{rc}\xspace}
\NewDocumentCommand \rcR			{s}
	{ \IfBooleanTF{#1}{ \vec{\mathcal{R}} }{ \mathcal{R} } }
\NewDocumentCommand \rcS			{s}
	{ \IfBooleanTF{#1}{ \vec{\mathcal{S}} }{ \mathcal{S} } }

\NewDocumentCommand \redinterp		{O{\rho}mo}
	{ \llbracket #2 \rrbracket_{#1} \IfValueT{#3}{^{(#3)}} }
\NewDocumentCommand \redsubst 		{O{\alpha}m}	{ [#1 \mapsto #2] }

\begin{document}

\title{Staying Productive Under the Palm Trees}
\subtitle{On Graded Coeffect Typing in the Tropical Semiring}

\author{%
	Rémy Cerda\inst{1}\orcidID{0000-0003-0731-6211} \and
	Ugo Dal Lago\inst{1,2}\orcidID{0000-0001-9200-070X}}
\institute{%
	Università di Bologna \and
	Centre INRIA d'Université Côte d'Azur \\
	\email{Remy.Cerda@math.cnrs.fr}, \email{ugo.dallago@unibo.it}}

\maketitle

\begin{abstract}
	We give a temporal reading of graded coeffects: \emph{tropical grades 	carry guards on infinite computation.} To illustrate it, we first define $\fmutypes$, a polymorphic calculus with guarded recursive types that guarantees 	productive infinite computation, in the form of hereditary head normalisation. The calculus supports stream computing \emph{via} the Scott encoding, and embeds Nakano's \emph{later} modality. We then introduce $\itypes$, a graded intersection type system with guarded infinitary types for the pure λ-calculus, whose typability allows for an exact characterisation of hereditary head normalisation. We prove that the characterising typability predicate is $\Pi^0_2$, matching the recursion-theoretic complexity of the characterised property.

\end{abstract}

\section{Introduction: A Temporal Reading of Graded Coeffects}
\label{sec:intro}
% !TeX root = ../main.tex
% !TeX spellcheck = en_GB

Computing with infinite data calls for a notion of progress that accounts for the information produced along the way, even when evaluation as a whole cannot terminate. In this paper, we investigate how graded type disciplines designed to control the use of function arguments can be given such an interpretation, by reading their grades as temporal requirements. Our starting point is the tropical semiring: we show how its grades carry the guards needed to control recursive and infinite type structure.

\subsubsection{Productivity and Guarded Recursion.}

For a computation producing an infinite object, \eg, a stream, the relevant requirement is that every finite prefix can be computed in finitely many steps. This property is \emph{productivity} \autocite{Dijkstra.80,Sijtsma.89,Kennaway.Klo.Sle.Vri.95}. The issue, therefore, is to distinguish computations that progressively produce an infinite object from those that diverge before producing some part of it.

Guarded recursion addresses this issue by separating what is available now from what is available later. Following Nakano's introduction of the \emph{later} modality \autocite{Nakano.00}, a rich body of work has developed type theories for productive recursive definitions \autocite{Birkedal.Mog.Sch.Sto.11,Atkey.McB.13,Moegelberg.14, Clouston.Biz.Gra.Bir.17,Guatto.18,Jaber.Rib.21}. Their basic idea is that a recursive definition may use its own result only through a temporal delay, ensuring that recursive unfoldings are \enquote{guarded}. For example, a guarded stream contains a head available now and a tail available one stage later.

To study productivity in the λ-calculus, we consider the notion of \emph{hereditary head normalisation}. Intuitively, head normalisation of a λ-term exposes a finite, stable prefix of information, consisting of abstractions and a head variable; its arguments are computations from which further information remains to be obtained. Hereditary head normalisation requires each of these computations to expose such a prefix in turn, and the same requirement to hold at every depth. Equivalently, the B\"ohm tree semantics of the λ-term contains no undefined node, so that successive head reductions can expose every finite part of the tree.

In this work, the connection with stream productivity is motivated and made concrete by using the Scott encoding \autocites{Curry.Hin.Sel.72}{Mogensen.92}. In this encoding, a constructor presents its fields to a continuation; for streams, ignoring type annotations and folding, we write
\[ \termstyle{cons}\;h\;t \eqdef \lambda c.c\,h\,t. \]
The head and tail computations are thus supplied without first traversing the rest of the stream; hereditary head normalisation controls the successive exploration of that tail, connecting our general operational property to finite observations of streams with ground elements.

\subsubsection{From Resource Bounds to Temporal Guards.}

Our approach to enforcing such delays starts from type disciplines that control how algorithms and subroutines use their arguments. Linear logic \autocite{Girard.87} makes this control possible by decomposing the ordinary function type $A\to B$ into the linear type $!A \lto B$: the exponential modality marks values that may be copied or erased. Bounded linear logic \autocite{Girard.Sce.Sco.92} takes this refinement further by introducing indexed, or \enquote{graded} exponentials $\bang[a]A$, whose indices bound resource usage, and allows for type-theoretic guarantees about computational complexity.

Once the exponential has been graded in this way, it is natural to consider what other information its grades might express. Graded coeffect systems generalise this idea \autocite{Brunel.Gab.Maz.Zda.14,Ghica.Smi.14,GKOBU.16, Katsumata.18,Orchard.Lie.Ead.19} by allowing grades to range over an arbitrary semiring. An argument carries a grade describing how it is required by its context, while the algebraic operations describe how these requirements compose. With natural-number grades, combining uses adds their requirements and nesting uses multiplies them, describing \emph{how many times} an argument will be used; changing the grading algebra permits other interpretations of the same type-theoretic structure.

We pursue this idea for the \emph{tropical semiring} of natural numbers extended with infinity,
$\Tropical \eqdef (\Nat\cup\{\infty\},\min,\infty,+,0)$.
Its additive operation is given by the minimum of its two arguments, while multiplication coincides with ordinary addition.
These operations suggest a temporal reading in which a grade records the earliest stage of a computation at which an argument is required. Grade $0$ represents an immediate requirement, while $\infty$ represents the absence of a requirement at any finite stage. Combining occurrences of the argument retains the earliest one, hence minimum; composing delays adds them. Since minimum is idempotent, repeated identical demands do not change their grade: temporal information is retained while usage multiplicity is discarded.

Accordingly, $\bang[a]A$ expresses a delay of $a$ stages. This gives us a way to impose guards on recursive types, by requiring recursive occurrences to lie behind strictly positive grades, \ie, to be used after a strictly positive delay. (Let us also emphasise that the time involved here is unrelated to execution time: it is given by progress through successive observations, and a single observation may require arbitrarily many reduction steps. Relating the two forms of information is a direction for further work.)

\subsubsection{Tropical Typing for Productive Computation.}

Putting this interpretation into practice, we define two type systems: a polymorphic system with guarded recursive types, $\fmutypes$, and a tropical intersection type system, $\itypes$. These constructions serve complementary purposes: we first establish that tropical guards suffice to guarantee productivity in the former system, then use the same guard mechanism to obtain an exact characterisation of productive pure $\lambda$-terms through the latter one.
\smallskip

In \cref{sec:g} we introduce system $\fmutypes$. We explain how the Scott encoding fits the tropical grading discipline and allows elementary causal stream programming. In particular, the observers $\termstyle{head} = \termstyle{obs}_0$, $\termstyle{tail}$, and $\termstyle{obs}_n$ make explicit how observation interacts with the guards carried by the types. Their types record the delay incurred when descending into a stream, and a finite-observation result establishes that every finite prefix of a well-typed stream with ground elements can be computed from its erased encoding.

Alongside these stream constructions, we show how Nakano's \emph{later} modality can be represented within $\fmutypes$, together with its \emph{next} operation, guarded fixed point, and functorial action. In this way, a familiar interface for guarded recursion is recovered from the underlying graded coeffect discipline. On the meta-theoretical side, we prove---by adapting Tait's reducibility method---that the type system guarantees productivity, formalised as hereditary head normalisation.
\smallskip

To obtain an exact characterisation in the pure λ-calculus, our second contribution introduces a variation on (non-idempotent) intersection types \autocite{Coppo.Dez.78,Gardner.94,Kfoury.00,Bucciarelli.Kes.Ven.17}, which we call \emph{tropical intersection types}. The key idea is to replace the classical notion of a set or multiset of types with that of a \enquote{tropical multiset} of types, where each type is annotated with temporal information inherited from the tropical grading. For example, an argument may be required at one type immediately and at another type two stages later; combining requirements retains the earliest demand for each type. Infinite types are then admitted under positive-grade guards, following the same principle as in $\fmutypes$.

The resulting system $\itypes$, presented in \cref{sec:it}, enjoys subject reduction and subject expansion. We prove that, under appropriate polarity restrictions on the empty tropical multiset, its typability precisely characterises hereditary head normalisation: whereas unrestricted typability characterises head normalisation (as usual), the polarity restriction ensures that the components exposed by head reduction are themselves observed. Finally, we prove, independently of completeness, that restricted typability lies at the $\Pi^0_2$ level of the arithmetical hierarchy, exactly matching the recursion-theoretic complexity of hereditary head normalisation.

\subsubsection{Related Work.}

The tropical semiring already appears among the possible grading algebras for coeffect typing \autocite{Brunel.Gab.Maz.Zda.14,Orchard.Lie.Ead.19}. For this instance, we develop the temporal interpretation into productivity guarantees and an intersection type characterisation. In relation to guarded recursion, we recover a familiar \emph{later}-based interface, while taking the guard condition from the grading algebra. Clock quantification and constant modalities provide richer ways to move between guarded and coinductive programming \autocite{Atkey.McB.13,Moegelberg.14,Clouston.Biz.Gra.Bir.17}; our focus is instead on the algebraic origin of the guards and their use in an exact operational characterisation. Our completeness result for $\itypes$ concerns productive pure λ-terms. Expressiveness of our systems as languages for defining stream functions remains a separate question.

Comonadic accounts of dataflow computation and tree relabelling already appear in \textcite{Uustalu.Vene.08}, but our use of grades is more closely connected to temporal coeffect analysis \autocite{Petricek.Orc.Myc.14}: in its dataflow interpretation, grades describe how far into the past a computation accesses its inputs, and maximum combines caching requirements. Our grades instead describe future demands, with minimum retaining the earliest required use. Thus, although both approaches describe temporal dependencies through coeffects, here these dependencies also govern recursive and infinite type formation and yield a productivity theorem.

Our use of intersection types builds on their ability to support exact operational characterisations of λ-terms \autocite{Coppo.Dez.78,Bucciarelli.Kes.Ven.17}. Another standard development is quantitative usage analysis, as in natural reduction types \autocite{Wright.Bak.93} and intersection types \autocite{Gardner.94,deCarvalho.17}; we keep turning these approaches into temporal analysis for further work. For hereditary head normalisation, \textcite{Tatsuta.08} established non-enumerability and gave a characterisation through typability at every type in a countable family. The closest comparison with our treatment of infinite types is the work of \textcite{Vial.17,Vial.21}, which likewise characterises hereditary head normalisation using infinitary intersection typing together with an approximability condition. We instead build temporal guardedness into type formation, while keeping the typing rules inductive, thereby providing a more direct way to control the infinitary structure.

Although we use Scott-encoded streams to illustrate this guard mechanism, the question of finite observable progress extends to other infinite objects. Guarded calculi also support branching structures such as infinite trees \autocite{Clouston.Biz.Gra.Bir.17}, while coinductive methods in exact real arithmetic describe digit-stream computations and read/write trees representing continuous functions \autocite{Berger.11}. Logical specifications can additionally track the input lookahead required to produce output digits \autocite{Koepp.Sch.23}. From our perspective, these connections suggest concrete applications in which to investigate how temporal grades relate successive observations to the inputs they require.

\section{Guarded Recursive Types \emph{via} Tropical Grades}
\label{sec:g}
% !TeX root = ../main.tex
% !TeX spellcheck = en_GB

In this section we present the essential features of $\fmutypes[\Tropical]$, a λ-calculus equipped with graded coeffects as well as polymorphic and recursive types. In application types and typing contexts, arguments are graded by an element of the tropical semiring~$\Tropical$; type recursion is allowed only under a non-zero grade.
After having defined the system~(\cref{g:ssec:thesystem})
we showcase how it allows to do some programming on Scott-encoded streams~(\cref{sec:g:sec:programming}) and
to embed the \emph{later} modality from guarded recursion~(\cref{sec:g:sec:later}).
Then we prove its main property, namely that typing guarantees the productivity of programs~(\cref{sec:g:sec:soundness}).

%****************************************************************************
\subsection{A Graded Language of Polymorphic, Recursive Programs}
\label{g:ssec:thesystem}

Let us start by recalling the definition of our graded algebra, namely the tropical semiring; we additionally equip the algebra with its canonical order relation.

\begin{definition}[Min-plus tropical semiring]
	The \emph{tropical semiring} is the semiring defined by
	$	\Tropical \eqdef 
		(\Nat\cup\{\infty\},\min,\infty,\mathord{+},0,\mathord{\geq})$.
\end{definition}

Then we first introduce the type language, which we deliberately keep very simple: we limit the constructors to graded application, polymorphism and recursion, which already allow to encode a wide range of features.

\begin{definition}[Types] \label{def:fmutypes}
	Let $\Atoms$ be a countable set of atoms. Types of $\fmutypes$
	are defined inductively as follows:
	\[	\fmutypes \bnfni A,B,\dots \bnfeq
	\alpha \bnfsep 
	\ito A B \bnfsep
	\fatype\alpha. A \bnfsep
	\rectype\alpha. A
	\qquad (\alpha \in \Atoms, a \in \Tropical) \]
	where the case $\rectype\alpha. A$ is restricted to
	atoms $\alpha$ and types $A$ such that
	$\bangindex{A}(\alpha) \neq 0$,
	where $\bangindex{A}$ is the map $\Atoms \to \Tropical$
	defined inductively by:
	\begin{gather*}
		\bangindex{\alpha}(\alpha) \eqdef 0 \qquad
		\bangindex{\alpha}(\beta) \eqdef \infty \qquad
		\bangindex{\ito A B}(\alpha) \eqdef
		\min(\bangindex{A}(\alpha) + a, \bangindex{B}(\alpha)) \\
		\bangindex{\fatype\alpha. A}(\beta) 
		\eqdef \bangindex{A}(\beta) \qquad
		\bangindex{\rectype\alpha. A}(\beta) \eqdef \bangindex{A}(\beta).
	\end{gather*}
	The operators $\fatype $ and $\mu$ act as binders for atoms, and types are taken modulo $\alpha$-equivalence.
\end{definition}

The simultaneous presence of polymorphic and recursive types may appear somewhat redundant. In fact, the design of our type system is not motivated by minimality; rather, its primary purpose is to introduce our grading discipline and to demonstrate its compatibility with standard type-theoretic constructs. We also note that the guardedness condition applies only to recursive types, whereas polymorphic types are left unconstrained.

The syntax of $\fmutypes$ terms is entirely standard, under the implementation choices that recursive types are treated in an iso-recursive fashion and that parametric polymorphism is reflected at the term level through the usual type abstraction and type application constructs.

\begin{definition}[Terms]
	Let $\Var$ be a countable set of variables.
	Terms are defined inductively~by:
	\begin{multline*}
		\fmuterms \bnfni M,N,\dots \bnfeq
		x \bnfsep 
		\abs x [A]. M \bnfsep
		MN \bnfsep
		\typeabs\alpha. M \bnfsep
		MA \bnfsep
		\fold M \bnfsep
		\unfold M
		\\ (x \in \Var, \alpha \in \Atoms).
	\end{multline*}
	The abstractions $\abs x[A].M$ and $\typeabs\alpha.M$ act as binders for term variables and atoms, respectively, and terms, as expected, are taken modulo the corresponding notion of $\alpha$-equivalence.
\end{definition}

We define substitution on terms and types as usual. The reduction relation $\bred$ on terms is the smallest congruence containing all instances of the following rules:
\begin{equation} \label{def:bred}
	(\abs x [A]. M)N \bred M \subst N \qquad
	(\typeabs\alpha. M)A \bred M \subst[\alpha] A \qquad
	\unfold (\fold M) \bred M.
\end{equation}
Noticeably, $\bred$ is confluent, which can be proved through standard arguments.

\smallskip

We are now ready to present the typing rules of our system. Any departures from standard formulations will be highlighted as they arise. Typing judgments are of the shape $\Gamma \vdash M:A$,
where the context $\Gamma$ is now a map $\Var \to \Tropical \times \fmutypes[\Tropical]$ whose \emph{support} $\support(\Gamma) \eqdef \set{x \in \Var}[\pi_0(\Gamma(x)) \neq \infty]$
is a finite set $\set{x_1,\dots,x_n}$.
We will present $\Gamma$ in the usual way, as a list
$x_1:\bang[a_1]C_1, \dots, x_n:\bang[a_n]C_n$.
The operations of $\Tropical$ are extended to contexts as follows:
if $\pi_1 \circ \Gamma = \pi_1 \circ \Gamma'$
(which will be implicitly assumed whenever we use this construction),
then 
\begin{align*}
	\min(\Gamma, \Gamma') : x & \mapsto \left( 
		\min \left(\pi_0(\Gamma(x)), \pi_0(\Gamma'(x))\right), 
		\pi_1(\Gamma(x)) \right) \\
	(a + \Gamma) : x & \mapsto (a + \pi_0(\Gamma(x)), \pi_1(\Gamma(x)))
\end{align*}
\ie, the operations are just mapped on the first component.
Typing rules are as follows, and essentially follow \textcite{Ghica.Smi.14}.

\begin{definition}[Typing in $\fmutypes$]
	A term $M \in \fmuterms$ has type $A \in \fmutypes$ in context $\Gamma$ if the judgement $\Gamma \vdash M:A$ is derivable (by a finite derivation $\derivD$) using the following rules:
	\begin{gather*}
		\begin{prooftree}
			\infer0[ax]{ \Gamma, x:\bang[0]A \vdash x:A }
		\end{prooftree}
		\\[\topsep]
		\begin{prooftree}
			\hypo{ \Gamma, x:\bang A \vdash M : B }
			\infer1[\lto_i]{ \Gamma \vdash \abs x[A]. M : \ito A B }
		\end{prooftree}
		\qquad
		\begin{prooftree}
			\hypo{ \Gamma \vdash M : \ito A B }
			\hypo{ \Gamma' \vdash N:A }
			\infer2[\lto_e]{ \min(\Gamma, a + \Gamma') \vdash MN:B  }
		\end{prooftree}
		\\[\topsep]
		\begin{prooftree}
			\hypo{ \Gamma \vdash M:B }
			\hypo{ \alpha \notin \fv(\Gamma) }
			\infer2[\fatype_i]{ \Gamma \vdash \typeabs\alpha.M : \fatype\alpha.B }
		\end{prooftree}
		\qquad
		\begin{prooftree}
			\hypo{ \Gamma \vdash M : \fatype\alpha.B }
			\infer1[\fatype_e]{ \Gamma \vdash MA : B \tsubst A }
		\end{prooftree}
		\\[\topsep]
		\begin{prooftree}
			\hypo{ \Gamma \vdash M : A \tsubst{\rectype \alpha.A} }
			\infer1[fold]{ \Gamma \vdash \fold M : \rectype\alpha.A }
		\end{prooftree}
		\qquad\qquad
		\begin{prooftree}
			\hypo{ \Gamma \vdash M : \rectype\alpha.A }
			\infer1[unfold]{ \Gamma \vdash \unfold M : A \tsubst{\rectype\alpha.A} }
		\end{prooftree}
	\end{gather*}
	We then write $\derivD \derives \Gamma \vdash M:A$,
	or simply $\derives \Gamma \vdash M:A$ to say that such a $\derivD$ exists.
\end{definition}

The rules \drule{ax}, \drule{\lto_i} and \drule{\lto_e} are standard for a graded coeffect system, while the other are perfectly in line with the usual ones for polymorphic and recursive types. In \drule{ax}, the variable is required to have grade $0$, namely the unit of the multiplicative structure of~$\Tropical$.
Importantly, context weakening is admissible thanks to the non-relevant \drule{ax} rule:

\begin{lemma}[Context weakening]
\label{lem:g:weakening}
	The following rule is admissible in $\fmutypes$:
	\[	\begin{prooftree}
			\hypo{ \Gamma \vdash M : A }
			\hypo{ x \notin \support(\Gamma) \cap \fv(B) }
			\infer2[weak]{ \Gamma, x:\bang[a]B \vdash M : A }
		\end{prooftree} \]
\end{lemma}

	\begin{proof}
	By induction on the derivation; when encountering the rule \drule{\lto_e}, induction is applied to the left premiss, in order not to increase the grade.
	\qed
	\end{proof}

\begin{restatable}[Subject reduction]{lemma}{sysGsubjectreduction}
\label{lem:g:sr}
	If $\derives \Gamma \vdash M:A$ and $M \bred M'$,
	then $\derives \Gamma \vdash M':A$.
\end{restatable}
The proof, detailed in \cref{app:g:subject-reduction}, adapts the standard arguments for System~$\mathsf{F}$.

\begin{remark}[Argument grades and $\eta$-expansion]
\label{rem:g:eta-grades}
	While general \enquote{argument anticipation}, in the form of a rule
	\[	\begin{prooftree}
		\hypo{ \Gamma \vdash M : \ito[a'] A B }
		\hypo{ a \le a' }
		\infer2{ \Gamma \vdash M : \ito[a] A B }
		\end{prooftree} \]
	would break subject reduction, this could actually be achieved modulo $\eta$-expansion:
	\[	\begin{smallprooftree}
		\hypo{ \Gamma \vdash M : \ito[a'] A B }
		\infer1[weak]{ \Gamma, x:\bang[a]A \vdash M : \ito[a'] A B }
		\infer0[ax]{ \Gamma, x:\bang[0]A \vdash x:A }
		\infer2[\lto_e]{ \Gamma, x:\bang[a]A \vdash Mx : B }
		\infer1[\lto_i]{ \Gamma \vdash \abs x[A].Mx : \ito[a] A B }
		\end{smallprooftree} \]
	where $a \le a'$ ensures that $\min(a, a' + 0) = a$.
\end{remark}

%****************************************************************************
\subsection{Productive Programming in $\fmutypes$}
\label{sec:g:sec:programming}

In this subsection we play a bit with our language, demonstrating its behaviours.
In particular, we will develop the paradigmatic example of programming on Scott-encoded streams.

Before starting, let us introduce some useful notations. First of all, for the sake of succinctness and clarity, we will write $\fd{M}$ for $\fold M$
and $\ufd{M}$ for $\unfold M$. Moreover, the type $\bang[0]A\lto B$ will be often written as $A\to B$. (By the way, notice that our language can embed System $\mathsf{F}$ by taking $A\to B$ as the underlying function space.) In particular, it is possible to encode finite and inductive types through the traditional impredicative encoding. For example, if we take $\booltype$ to be the standard type $\fatype\alpha.\alpha\to\alpha\to\alpha$ of booleans, we can write the truth values as follows:. 
\begin{gather*}
	\trueterm \eqdef \typeabs \alpha. \abs x[\alpha]. \abs y[\alpha]. x : \booltype \qquad
	\falseterm \eqdef \typeabs \alpha. \abs x[\alpha]. \abs y[\alpha]. y : \booltype.
\end{gather*}	

\subsubsection{Recursion.}

While grades do not interact with polymorphism, they play a crucial role in the construction of recursive types. If we take the two types $A \eqdef \rectype \beta.\ito[1]{\beta}{\alpha}$ and $F \eqdef \ito[1]{\alpha}{\alpha}$, we can use them to assign a type to Curry's fixed-point combinator, which here becomes the term:
\[
	Y \eqdef \typeabs \alpha. \abs f [F]. 
	(\abs x[A]. f(\ufd{x}x))
	(\fd{\abs x[A]. f(\ufd{x}x)})
\]
(we will use it in an untyped fashion to lighten the notations). The term $Y$ receives the type $\fatype\alpha.(\ito[1]{\alpha}{\alpha})\to\alpha$ as witnessed by the following derivation (where some of the deduction steps have been contracted for the sake of conciseness):
\begin{equation} \label{eq:fmutypes-derivation-y}
	\begin{smallprooftree}[center]
	\infer0{ f:\bang[0]F \vdash f:F }
	\infer0{ x:\bang[0]A \vdash x:A }
	\infer1{ x:\bang[0]A \vdash \ufd x : \ito[1]{A}{\alpha} }
	\infer0{ x:\bang[0]A \vdash x:A }
	\infer2{ x:\bang[0]A \vdash \ufd xx : \alpha }
	\infer2{ f:\bang[0]F, x:\bang[1]A \vdash f(\ufd xx) : \alpha }
	\infer1{ f:\bang[0]F \vdash \abs x[A].f(\ufd xx): \ito[1]{A}{\alpha} }
	\hypo{\vdots}
	\infer1{ f:\bang[0]F \vdash \abs x[A].f(\ufd xx): \ito[1]{A}{\alpha} }
	\infer1{ f:\bang[0]F \vdash 
		\fd{\abs x[A].f(\ufd xx)}: A }
	%	\rewrite{\hspace*{-.5cm}\box\treebox}
	\infer2{ f:\bang[0]F \vdash 
		(\abs x[A].f(\ufd xx))(\fd{\abs x[A].f(\ufd xx)}):\alpha }
	\infer1{ \vdash : \abs f [F]. (\abs x[A]. f(\ufd{x}x))(\fd{\abs x[A]. f(\ufd{x}x)}):\ito[0]{(\ito[1]{\alpha}{\alpha})}{\alpha} }
	\infer1{\vdash Y: \fatype\alpha.\ito[0]{(\ito[1]{\alpha}{\alpha})}{\alpha}}
\end{smallprooftree}
\end{equation}
Observe how the type $A$ is well-formed precisely thanks to the use of grade $1$. This constraint is then implicitly expressed in the type we assign to $Y$, in which the grade $1$ cannot be reduced to $0$. By this typing, we automatically avoid the possibility of applying it \eg to the identity $\abs x[A].x$, because the latter term can be assigned the type $A \to A$, but \emph{none} of the types $\ito[a]{A}{A}$ with $a \geq 1$. The \enquote{type} $B \eqdef \rectype \beta.\ito[0]{\beta}{\alpha}$, instead, is not guarded, hence ill-formed. If it were legal, one could easily type
$\Omega\eqdef (\abs x[A].\ufd xx)(\fd{\abs x[A].\ufd xx})$, thus undermining any hope of establishing productivity properties.

\subsubsection{Stream Programming.}

We shall now use this discipline to represent coinductive data streams and compute on it. For this purpose, we use the Scott encoding
\autocites[\S~13.C]{Curry.Hin.Sel.72}{Mogensen.92}: a constructor
presents its immediate components to a continuation, without first
traversing the rest of the object.
In the case of streams, the continuation receives the head and tail;
the latter retains its stream type and appears as
a grade-$1$ argument of the continuation, making the guard
explicit in the constructor interface.
(In this regard the standard Church encoding cannot be used, as its constructor folds the head and the result of the recursive folding of the tail, so that the tail would be used with grade~$0$.)

\begin{definition}[Scott-encoded Streams]
	Given a type $A$, the type $\strtype[A]$ of (infinite) streams of elements of type~$A$ is defined by
	$\strtype[A] \eqdef \rectype \beta. \fatype \gamma.
		\bang[0](\bang[0]A \lto \bang[1]\beta \lto \gamma) \lto \gamma$.
\end{definition}

The type is well-formed because the recursive occurrence in
the tail lies behind grade $1$.
%Observe that the type of finite or infinite streams could as well be defined, by
%$\strtype*[A] \eqdef \rectype \beta. \fatype \gamma.
%	\bang[0]\gamma \lto \bang[0](\bang[0]A \lto \bang[1]\beta \lto \gamma) 
%	\lto \gamma$.
%
To make the encoding explicit,
fix $A$ and denote by $C_\gamma \eqdef A\to\ito[1]{\strtype[A]}{\gamma}$ the type of continuations.
The stream constructor is
\[	\termstyle{cons}_A
	\eqdef
	\abs{h}[A].\abs{t}[\strtype[A]].
	\fd{\typeabs\gamma.\abs{c}[C_\gamma].cht}
	: \ito[0] A {\ito[1]{\strtype[A]}{\strtype[A]}}. \]
Observe that head normalisation makes the head $h$ observable:
\[	\ufd{\termstyle{cons}_A\, h\, t} \bred*
	\typeabs\gamma. \abs c[C_\gamma]. cht. \]
Productivity, \ie the fact that each element of the stream will become observable through reduction, hence corresponds to recursive head normalisation in the tail $t$: this is \emph{hereditary head normalisation}, to be formally introduced in \cref{def:productive}.

\begin{example}[Head and tail, first attempt]
\label{exa:g:head-tail}
	The Scott encoding gives direct access to the head of a stream:
	\[
	  \termstyle{head}_A \eqdef
	  \abs s[\strtype[A]].
	  \ufd{s}\,A\,(\abs h[A].\abs t[\strtype[A]].h)
	  : \strtype[A] \to A.
	\]
	The analogous projection of the tail would use
	$\abs h[A].\abs t[\strtype[A]].t$ as its continuation.
	However, this continuation cannot receive the required type
	$\ito[0] A {\ito[1]{\strtype[A]}{\strtype[A]}}$:
	returning $t$ directly uses it
	at grade~$0$, whereas this type requires grade~$1$.
	We can nevertheless pass the tail to a continuation which
	accepts it at grade~$1$:
	\[
	  \termstyle{tail}^{\mathrm{cps}}_A
	  \eqdef
	  \abs s[\strtype[A]].
	  \typeabs\delta.
	  \abs d[\ito[1]{\strtype[A]}{\delta}].
	  \ufd{s}\,\delta\,(\abs h[A].\abs t[\strtype[A]].dt)
	  : \strtype[A] \to \fatype\delta.
	    (\ito[1]{\strtype[A]}{\delta}) \to \delta.
	\]
	These operations satisfy
	$\termstyle{head}_A(\termstyle{cons}_A  h t) \bred* h$ and
	$\termstyle{tail}^{\mathrm{cps}}_A (\termstyle{cons}_A h t)\,B\,d \bred* dt$, \ie the head can be returned directly, while access to the
	tail retains its guard through a continuation.
\end{example}

%A polymorphic function outputting a constant infinite stream is the following one:
%\[	\termstyle{const} \eqdef \typeabs \alpha. \abs x[α]. Y \left(
%	\abs f[\strtype[\alpha]]. \fd{ \typeabs\gamma.
%	\abs c[\ito[0] \alpha { \ito[1] {\strtype[\alpha]} \gamma }]. c x f }
%	\right) : \fatype \alpha. \alpha \to \strtype[\alpha]. \]

\begin{example}[Mapping on streams]
	The following function applies uniformly an input function of type $\alpha\to\alpha$ to the input stream in $\strtype[\alpha]$:
	\begin{multline*}
	\termstyle{map} \eqdef \typeabs α. Y \left( 
		\begin{multlined}
		\abs m[(α \to α) \to \ito[1] {\strtype[α]} {\strtype[α]}]. 
		\abs f[\ito[0]{α}{α}]. \abs s[\strtype[α]]. 
		\ufd{s}\, \strtype[α]  \\
		\left( 
			\abs h[α]. \abs t[\strtype[α]]. 
			\fd{ \typeabs\gamma. 
			\abs c[\ito[0] \alpha { \ito[1] {\strtype[\alpha]} \gamma }]. 
			c (fh) (mft) } \right)
		\end{multlined}
		\right) \\ 
	: \fatype α. (α \to α) \to \strtype[α] 
	\to \strtype[α].
	\end{multline*}
\end{example}

A slight variation on the theme would be the function $\termstyle{map_2}$ 
which instead takes as input a function of type $\alpha\to\alpha\to\alpha$ 
and applies it to the elements of the input stream 
at positions $2n$ and $2n+1$, 
then returning the result in the output stream at position $n$. 
Noticeably, $\termstyle{map_2}$ is not \emph{causal}: 
outputting the first $n$ elements of the output streams requires 
knowing the input streams at positions strictly greater than $n$. 
The fact that $\termstyle{map_2}$ cannot be represented in $\fmutypes$ 
as a function acting on $\strtype[\alpha]$ to $\strtype[\alpha]$ is entirely natural, 
since the type discipline enforces causality at the level of grades: 
a box of type $\bang[a]A$ cannot influence the contents 
of a box of type $\bang[b]B$ whenever $b<a$, 
\ie, one cannot \enquote{go back in time}. 
This means that, in terms of expressiveness, 
many among the so-called guarded recursion disciplines, 
which permit the definition of non-causal functions, 
are strictly more expressive than $\fmutypes$. 
On the other hand, the \emph{later} modality can be faithfully captured, 
as we will show shortly. 
There is, however, a way to represent non-causal stream functions in $\fmutypes$. 
The idea is to allow the guard $\bang[1]$ occurring in the type $\strtype[A]$ to be 
replaced by an arbitrary grade $a>1$, yielding the following type:
\[
\strtype[A](a) \eqdef \rectype \beta. \fatype \gamma.
\bang[0](\bang[0]A \lto \bang[a]\beta \lto \gamma) \lto \gamma
\]
This extension would make it possible for the domain and codomain of stream functions 
to have different stream types, thereby enabling the encoding of non-causal definitions 
without compromising causality at the level of grades. For instance, $\termstyle{map}_2$ 
could in this way be assigned the type $\fatype α. (α \to α\to α) \to\strtype[\alpha](1)\to\strtype[\alpha](2)$.

%****************************************************************************
\subsection{Embedding the \emph{Later} Modality}
\label{sec:g:sec:later}

We show that the $\lm$ (\enquote{later}) modality introduced by \textcite{Nakano.00}, 
and originally proposed as a simple mechanism for guarded recursion, 
admits a natural representation within~$\fmutypes$. 
Nakano's $\lm$ modality can be applied to any type, 
whereas graded comonads, at least in the form considered here, 
implicitly appear to the left of a function type. 
To recover the full generality of the $\lm$ modality,
we resort to a form of double-negation translation:
\[
\lm A\eqdef (\bang[1] A\lto\bot)\to\bot
\]
where $\bot$ can be taken as, \eg, $\fatype\alpha.\alpha$. Informally, this type tells us that a term of type $\lm A$ has head normal form $\lambda x.xM$, where $x$ has type $\bang[1]A\lto\bot$. Consequently, $M$ becomes available \emph{at the next time step}. In this way, it is possible to capture all the constructions proposed by Nakano in his work, which we briefly illustrate in the following.

\subsubsection{Basic Operations for Guarded Recursion.}

We present three basic operations
of guarded recursion, already present in
Nakano's work and standard in later presentations
\autocite{Clouston.Biz.Gra.Bir.17}.
(The corresponding typing derivations are given in \cref{app:g:ssec:later-examples}.)

\begin{itemize}
\item The \emph{next} combinator introduces the later modality
	by passing its argument to a continuation:
	\[
	  \termstyle{next}_A
	  \eqdef
	  \abs x[A].\abs c[\ito[1]{A}{\bot}].cx
	  : A\to\lm A.
	\]
	This typing uses admissible weakening
	(\cref{lem:g:weakening}).
	The same term also has the stronger type
	$\ito[1]{A}{\lm A}$, since $x$ is used only as
	the grade-$1$ argument of the continuation.
\item This stronger typing lets us recover the guarded fixed-point
	type of Nakano's Example~1
	(which, in other words, encodes the Löb rule for $\lm$):
	\[
	  \termstyle{gfix}_A
	  \eqdef
	  \abs f[\lm A\to A].
	  Y(\abs x[A].f(\termstyle{next}_A x))
	  : (\lm A\to A)\to A.
	\]
	Indeed, the function supplied to $Y$ has type
	$\ito[1]{A}{A}$: its recursive argument is passed to $f$
	through $\termstyle{next}_A$.
\item Finally, a function and its argument may both be available
	under the later modality. Guarded application, sometimes written
	$\circledast$, is represented here by the combinator:
	\begin{multline*}
	  \termstyle{gapp}_{A,B}
	  \eqdef
	  \abs f[\lm(A\to B)].
	  \abs u[\lm A].
	  \abs k[\ito[1]{B}{\bot}].
	  f(\abs g[A\to B].u(\abs x[A].k(gx))) \\
	  :\lm(A\to B)\to\lm A\to\lm B.
	\end{multline*}
\end{itemize}
Not only do these combinators inhabit the expected types; they also enjoy the expected behaviour. Writing $\equiv$ for convertibility generated by $\bred$, they indeed satisfy the equations:
\begin{align*}
  \termstyle{gfix}_A f
  &\equiv
    f(\termstyle{next}_A(\termstyle{gfix}_A f)),
  &
  \termstyle{gapp}_{A,B}
    (\termstyle{next}_{A\to B} f)
    (\termstyle{next}_A M)
  &\equiv
    \termstyle{next}_B(fM).
\end{align*}
following  from the definitions and the fixed-point identity of $Y$.

\subsubsection{Back to Observing Streams.}

In \cref{exa:g:head-tail} we were able to return directly the head of stream, whereas the tail was only passed to a continuation accepting it at grade~$1$. The later modality now lets us give the tail a guarded result type, without exposing this continuation in its interface.

Fix $A$ again.
We retain $\termstyle{head}_A : \strtype[A] \to A$ from that example
and define
\[
  \termstyle{tail}_A
  \eqdef
  \abs s[\strtype[A]].
  \ufd{s}\,(\lm \strtype[A])\,
    (\abs h[A].\abs t[\strtype[A]].\termstyle{next}_{\strtype[A]} t)
  : \strtype[A] \to \lm \strtype[A].
\]
The stronger typing
$\termstyle{next}_{\strtype[A]} : \ito[1]{\strtype[A]}{\lm\strtype[A]}$
gives the continuation the required type
$A \to \ito[1]{\strtype[A]}{\lm\strtype[A]}$.
As expected, the reduction
$	\termstyle{tail}_A(\termstyle{cons}_A ht)
	\bred* \termstyle{next}_{\strtype[A]} t$
holds.
%(The term
%$\abs s[S].\termstyle{tail}^{\mathsf{cps}}_A\,s\,\bot$
%gives another tail of type $S\to\lm S$, with the same
%constructor equation.)

We can then construct observations at arbitrary finite depths,
with their types recording the number of tail accesses.
For each fixed $n\in\Nat$, write
$\lm^0 A \eqdef A$ and
$\lm^{n+1} A \eqdef \lm (\lm^n A)$, and define
$\termstyle{obs}_{A,n} : \strtype[A] \to \lm^n A$ by induction:
\begin{gather*}
  \termstyle{obs}_{A,0} \eqdef \termstyle{head}_A,
\qquad
  \termstyle{obs}_{A,n+1} \eqdef
    \abs s[\strtype[A]]. \termstyle{gapp}
      (\termstyle{next}\,\termstyle{obs}_{A,n})
      (\termstyle{tail}_A s).
\end{gather*}
(The type instances of $\termstyle{next}$ and
$\termstyle{gapp}$ are left implicit.)
These constructors verify
$\termstyle{obs}_{A,n+1}(\termstyle{cons}_A ht)
	\equiv \termstyle{next}_{\lm^n A} (\termstyle{obs}_{A,n} t)$,
\ie, they access to the $n$th element of their argument stream
in a guarded way, under $n$ $\mathtt{next}$ combinators.

%****************************************************************************
\subsection{Typability Implies Productivity}
\label{sec:g:sec:soundness}

In this section, we present the essential ingredients of the proof that the evaluation of $\fmutypes$~terms is a \emph{productive} process, which, as already explained, is formalised through the notion of \emph{hereditary head normalisation}. The proof is based on reducibility rather than denotational semantics. 

We begin by introducing the standard notions of \emph{head reduction}, \emph{head normal form} and \emph{head normalisation}, from which we build hereditary head normalisation.
\begin{definition} \label{def:head}
	The set $\HNF$ of terms in \emph{head normal form} is defined
	by induction by:
	\begin{align*}
		x \in \HNF \\
		\abs x[A]. M \in \HNF & \text{ if } M \in \HNF &
		MN \in \HNF & \text{ if } M \in \HNF 
			\text{ and } M \neq \abs x[A].P \\
		\typeabs α.M \in \HNF & \text{ if } M \in \HNF &
		MA \in \HNF & \text{ if } M \in \HNF
			\text{ and } M \neq \typeabs\alpha.P \\
		\fold M \in \HNF & \text{ if } M \in \HNF&
		\unfold M \in \HNF & \text{ if } M \in \HNF
			\text{ and } M \neq \fold P.
	\end{align*}			
	The partially functional sub-relation $\hred$ of $\bred$ can be defined as usual, by restricting reduction to happen on the left-hand-side of applications. We say that a term \emph{has a head normal form} or that if is \emph{head normalising} if it reduces to a head normal form through $\hred$. Head normalising terms form the set $\HN$.
\end{definition}

%\begin{definition} \label{def:head}
%	We simultaneously define:
%	\begin{itemize}
%		\item the subset of $\fmuterms$ of all terms in \emph{head normal form}
%		(\HNF),
%		\item the functional sub-relation $\hred$ of $\bred$,
%		called \emph{head reduction},
%		relating terms not in \HNF to (arbitrary) terms,
%	\end{itemize}
%	by induction on terms:
%	\begin{itemize}
%		\item $x$ is in \HNF;
%		\item $\abs x[A]. M$ is in \HNF whenever $M$ is; \\
%		otherwise $M \hred M'$ and we set $\abs x[A].M \hred \abs x[A].M'$;
%		\item $MN$ is in \HNF whenever $M$ is in \HNF 
%		and $M \neq \abs x[A].P$; \\
%		otherwise, if $M = \abs x[A].P$ we set $MN \hred P\subst N$; \\
%		otherwise $M \hred M'$ and we set $MN \hred M'N$;
%		\item $\typeabs α. M$ is in \HNF whenever $M$ is; \\
%		otherwise $M \hred M'$ and we set $\typeabs α.M \hred \typeabs α.M'$;
%		\item $MA$ is in \HNF whenever $M$ is in \HNF 
%		and $M \neq \typeabs\alpha.P$; \\
%		otherwise, if $M = \typeabs\alpha.P$ we set $MA \hred P\tsubst A$; \\
%		otherwise $M \hred M'$ and we set $MA \hred M'A$;
%		\item $\fold M$ is in \HNF whenever $M$ is; \\
%		otherwise $M \hred M'$ and we set $\fold M \hred \fold M'$;
%		\item $\unfold M$ is in \HNF whenever $M$ is in \HNF 
%		and $M \neq \fold P$; \\
%		otherwise, if $M = \fold P$ we set $\unfold M \hred P$; \\
%		otherwise $M \hred M'$ and we set $\unfold M \hred \unfold M'$.
%	\end{itemize}
%	We say that a term \emph{has a \HNF} of is \emph{head normalising}
%	if it reduces to a \HNF through $\bred$.
%\end{definition}

Observe that an application $MN$ is in head normal form (or not) irrespective of $N$. If we insist on the latter to be itself a normal form, we obtain the following notion, which captures the informal notion of productivity.

\begin{definition} \label{def:productive}
	A term $M \in \fmuterms$ is \emph{hereditarily head normalising} if it head reduces to a term in $\HNF$ such that the arguments to the head variable themselves are, coinductively so. More rigorously, the set $\HHN$ of hereditarily head normalising terms is the biggest set of terms such that whenever $M\in\HHN$ one of the following conditions hold:
	\begin{align*}
		M &\hred* x \\
		M &\hred* \abs x[A]. P
			\text{ and } P\in\HHN &
		M &\hred* PQ
			\text{, } P\in\HHN
			\text{, } P \neq \abs x[A]. P'
			\text{ and } Q\in\HHN\\
		M &\hred* \typeabs α.P
			\text{ and } P\in\HHN &
		M &\hred* PA
			\text{, } P\in\HHN
			\text{ and } P \neq \typeabs α. P'\\
		M &\hred* \fold P
			\text{ and } P\in\HHN &
		M &\hred* \unfold P
			\text{, } P\in\HHN
			\text{ and } P \neq \fold P'.	
	\end{align*}
\end{definition}

In λ-calculus jargon, this amounts to saying that
the Böhm tree of $M$ contains no undefined subterm $\bot$. In the rest of this section we give some details about the proof of the following Theorem; some proofs are delayed to \cref{app:g:ssec:productivity}.

\begin{theorem} \label{thm:typable-productive}
	If $M \in \fmuterms$ is typable then $M\in\HHN$.
\end{theorem}

The proof is by reducibility candidates,
as in \eg \textcite{Girard.Tay.Laf.89},
along the following sequence of definitions and lemmas.
We use the \enquote{forgetting trick}
from \textcite{Tait.75,Mitchell.86},
\ie the target of the reducibility interpretation
is the pure λ-calculus.

\begin{definition}[Pure λ-calculus]
	The set of \emph{pure λ-terms} is defined inductively by:
	\[	\pureterms \bnfni T,U,\dots \bnfeq x \bnfsep \abs x.T \bnfsep TU
	\qquad (x \in \Var). \]
	It is equipped with the reduction $\bred$, defined to be
	the smallest congruence containing $(λx.T)U \bred T \subst U$.
	\emph{Head normal forms} and 
	the relation $\hred$ of \emph{head reduction}
	are defined like in the first three clauses of \cref{def:head}.
	The set of all head normalising pure λ-terms
	is also denoted by $\HN$.
\end{definition}

\begin{definition} \label{def:rc}
	A \emph{reducibility candidate}, \rc in short,
	is a set $\rcR \subseteq \pureterms$ such that:
	\begin{enumerate}
		\item \label{def:rc:clause1} $\rcR \subseteq \HN$,
		\item \label{def:rc:clause2}
		if $T \in \rcR$ and $T \bred* T'$, then $T' \in \rcR$,
		\item \label{def:rc:clause3}
		if $T \neq \abs x.U$ and
		either $T$ is in $\HNF$ or $T \hred T' \in \rcR$,
		then $T \in \rcR$.
	\end{enumerate}
\end{definition}

As usual we also write $\rcR \to \rcS \eqdef
\set{ T \in \pureterms }[ \forall U \in \rcR,\ TU \in \rcS ]$,
which is a \rc as soon as $\Var \subseteq \rcR \subseteq \pureterms$
and $\rcS$ is a \rc.

\begin{definition} \label{def:redinterp}
	Given an \emph{environment} $\rho$ mapping every atom to a \rc,
	the \emph{reducibility interpretation} $\redinterp A$
	of a type $A$ is the \rc defined inductively by:
	\begin{align*}
		\redinterp{ \alpha } & \eqdef \rho(\alpha), \\
		\redinterp{ \ito[0] A B} & \eqdef  \redinterp A \to \redinterp B, \\
		\redinterp{ \ito A B} & \eqdef  \pureterms \to \redinterp B,
		& \text{for $a \neq 0$,} \\
		\redinterp{ \fatype \alpha.A } & \eqdef \bigcap_{\rcR \text{ a \rc}} 
		\redinterp[\rho \redsubst \rcR] A, \\
		\redinterp{ \rectype \alpha.A }
		& \eqdef \redinterp {A \tsubst{\rectype α.A}},
	\end{align*}
	where $\rho \redsubst \rcR$ is the environment mapping
	$α$ to $\rcR$ and any other atom $\beta$ to $\rho(\beta)$.
\end{definition}

The definition of $\redinterp{ \rectype \alpha.A }$ does not induce any loop,
because the condition on recursive types in \cref{def:fmutypes}
ensures that $α$ occurs in $A$ only under some $\bang[a]$ with $a \neq 0$,
hence $\redinterp{ \rectype \alpha.A }$ won't be computed again
when unfolding $\redinterp {A \tsubst{\rectype α.A}}$.
The fact that $\redinterp A$ is a \rc can be proved by induction.

\begin{definition}
	A forgetful translation $\forgetFmu{-} : \fmuterms \to \pureterms$
	is defined in the natural way:
	\begin{gather*}
		\forgetFmu x \eqdef x \qquad
		\forgetFmu{\abs x[A].M} \eqdef \abs x.\forgetFmu M \qquad
		\forgetFmu{MN} \eqdef \forgetFmu M \forgetFmu N \\
		\forgetFmu{\typeabs α.M} \eqdef \forgetFmu M \qquad
		\forgetFmu{MA} \eqdef \forgetFmu M \qquad
		\forgetFmu{\fold M} \eqdef \forgetFmu M \qquad
		\forgetFmu{\unfold M} \eqdef \forgetFmu M.
	\end{gather*}
\end{definition}

\begin{restatable}{lemma}{sysGiftypabletheninrc}
\label{lem:if-typable-then-in-rc}
	If $\Gamma \vdash M : A$,
	then $\forgetFmu M \in \redinterp[\rho_\HN] A$,
	where $\rho_\HN$ is the environment mapping all atoms to $\HN$.
\end{restatable}

\begin{restatable}{lemma}{sysGforgethn} \label{lem:forget-hn}
	If $\forgetFmu M \in \HN$ then $M \in \HN$.
\end{restatable}

	\begin{proof}[Proof of \cref{thm:typable-productive}]
	\label{thm:typable-productive:proof}
	Let $M$ be a typable term,
	\ie there is a derivation $\Gamma \vdash M : A$.
	By \cref{lem:if-typable-then-in-rc}
	and \cref{def:rc}, \cref{def:rc:clause1},
	$\forgetFmu M \in \redinterp[\rho_\HN] A \subseteq \HN$.
	By \cref{lem:forget-hn}, $M$ is head normalising.
	By subject reduction
	the head normal form of $M$ is typable,
	and in particular all its direct subterms
	(in the sense of \cref{def:productive}) are typable,
	hence productive by coinduction.
	Thus $M$ is productive.
	\qed
	\end{proof}

To relate hereditary head normalisation to finite stream observations,
we make the consequence for Scott-encoded streams explicit
using the constructions from \cref{sec:g:sec:later}.
We use $\booltype$ as an example of a ground type.

\begin{corollary}[Finite stream observation]
	\label{cor:g:finite-observation}
	For all term $M$ such that $\vdash M:\strtype[\booltype]$
	and for all $n\in\Nat$, 
	there is \(b\in\{\termstyle t,\termstyle f\}\) such that
	\(\termstyle{obs}_{\booltype,n} M \equiv
	\termstyle{next}^{n} b\)
	(the type index of $\termstyle{next}$ is left implicit).
\end{corollary}

\section{Tropical Intersection Types Characterising Productivity in the λ-Calculus}
\label{sec:it}
% !TeX root = ../main.tex
% !TeX spellcheck = en_GB

In \cref{sec:g}, we used tropical grades to construct a type discipline $\fmutypes$ guaranteeing hereditary head normalisation. We now carry the same guard mechanism to intersection types, with the aim of obtaining an exact characterisation of productive pure $\lambda$-terms.
The resulting system $\itypes$ admits guarded infinitary types (\cref{sec:it:sec:thesystem}).
Under a suitable polarity restriction, its typable terms are precisely those in $\HHN$ (\cref{sec:it:sec:soundcomplete,sec:it:sec:complete-proof}).
The system turns out to be optimal, in the sense that typing in our system can be directly proved to be a $\Pi_0^2$ predicate, just as the target set $\HHN$ (\cref{sec:it:sec:complexity}).

%****************************************************************************
\subsection{The Typing System}
\label{sec:it:sec:thesystem}

Classical non-idempotent intersection types associate to each argument position (and variable in the typing context) not a single type, but \emph{a multiset of types}, allowing to derive, \eg:
\[	\vdash λx.x(xx) : 
	[[\alpha] \lto \alpha, [\alpha] \lto \alpha, \alpha] \lto \alpha. \]
In this example, the three types in the multiset
$[[\alpha] \lto \alpha, [\alpha] \lto \alpha, \alpha]$
correspond to the three occurrences of the bound variable $x$. Put otherwise, any argument of $λx.x(xx)$ will be used twice with type $[\alpha] \lto \alpha$ and once with type $\alpha$, which corresponds to the definition of the above multiset as the map:
\[	([\alpha] \lto \alpha) \mapsto 2, \qquad
	\alpha \mapsto 1, \qquad
	\sigma \mapsto 0 \text{ for any other $\sigma$.} \]
This behaviour expresses the widely explored connexion between intersection typing and \emph{resource usage} \autocite{Bucciarelli.Kes.Ven.17}.
In this article, we focus instead on \emph{resource availability}: the main idea we are exploring is an intersection type system based upon multisets that do not map types to \emph{quantities} (in $\Nat$) any more, but to \emph{delays} (in $\Tropical$). Let us now formally define such multisets, and then the resulting notion of intersection type.

\begin{definition}[$\Tropical$-Multisets]
\label{def:i:tropical-msets}
	For any set $X$, a \emph{$\Tropical$-multiset} on $X$
	is a map $\ms m : X \to \Tropical$ that is \emph{finitary},
	\ie, for all $d \neq \infty$,
	the set $\set{x\in X}[\ms m(x) = d]$ is finite.
	We denote by $\msets[\Tropical]{X}$ the set of all $\Tropical$-multisets.
	
	Given finite families $\{x_i\}_{i \in I}$ in $X$ and
	$\{a_i\}_{i \in I}$ in $\Tropical$,
	we write $[\gr{x_i}{a_i}]_{i \in I}$
	for the $\Tropical$-multiset $\ms m$ such that
	$\ms m(x) \eqdef \min_{i \in I, x_i = x} a_i$.
	In the following we always try to keep the $x_i$ pairwise distinct,
	so that the definition boils down to $\ms m(x_i) \eqdef a_i$,
	and $\ms m(x) \eqdef \infty$ for $x \notin \set{x_i}[i \in I]$.
	When $I = \set{1,\dots,n}$ we may also
	denote $\ms m$ by $[\gr{x_1}{a_1}, \dots, \gr{x_n}{a_n}]$.
	In particular $[]$ is the \emph{empty $\Tropical$-multiset}.
\end{definition}

Notice also that $[]$ is the constant map with value $\infty$ (not with value~$0$).
The operations of $\Tropical$ induce an internal sum
and an external product on $\Tropical$-multisets:
\[	\min (\ms m, \ms n)(x) \eqdef \min(\ms m(x), \ms n(x)) \qquad
	(a + \ms m)(x) \eqdef a + \ms m(x). \]

For example, let $x$ and $y$ be distinct elements of $X$ and consider the $\Tropical$-multisets on $X$ defined by
$\ms m\eqdef[\gr{x}{0},\gr{y}{2}]$ and
$\ms n\eqdef[\gr{x}{3},\gr{y}{1}]$.
Thus $\ms m(x)=0$, $\ms m(y)=2$, and
$\ms m(z)=\infty$ for every other $z \in X$.
Combining and shifting these $\Tropical$-multisets gives
$\min(\ms m, \ms n) = [\gr{x}{0},\gr{y}{1}]$ and
$2 + \ms m = [\gr{x}{2},\gr{y}{4}]$.

\begin{definition}[$\Tropical$-Intersection Types]
\label{def:itypes}
	Let $\Atoms$ be a countable set of atoms.
	The sets $\itypes$ of \emph{$\Tropical$-intersection types}
	and $\itypes! \subset \msets[\Tropical]{\itypes}$
	of \emph{$\Tropical$-intersection multi-types}
	are defined by mutual induction and coinduction
	by the following system of formation rules:
	% Previous version of the rules, see commit 
	% 9cc090d458b5ef2b679b14b0e9b60b9979a5c231
	\begin{gather*}
		\begin{prooftree}
		\hypo{ α \in \Atoms }
		\infer1[\Atoms]{ α \in \itypes }
		\end{prooftree}
	\qquad\qquad
		\begin{prooftree}
		\hypo{ \ms{m} \in \itypes! }
		\hypo{ \tau \in \itypes }
		\infer2[\lto]{ \ms{m} \lto \tau \in \itypes }
		\end{prooftree}
	\\[\topsep]
		\begin{prooftree}[center]
		\hypo{ \left[ \gr {(\sigma_i \inlater{a_i} \itypes)} 
			{a_i} \right]_{i \in I} }
		\infer1[!]{ [\gr{\sigma_i}{a_i}]_{i \in I} \in \itypes! }
		\end{prooftree}
	\qquad
		\begin{prooftree}[center]
		\hypo{ \sigma \in \itypes }
		\hypo{ a > 0 }
		\infer[double]2[coind]
			{ \sigma \inlater{a} \itypes }
		\end{prooftree}
	\qquad
		\begin{prooftree}[center]
		\hypo{ \sigma \in \itypes }
		\infer1[ind]{ \sigma \inlater{a} \itypes }
		\end{prooftree}
	\end{gather*}
	where single-bar rules are treated inductively
	whereas double-bar rules are treated coinductively.
	
	Notice that the formation rule \drule{!} has a $\Tropical$-multiset
	of hypotheses. Notice also that these hypotheses feature relations $\inlater{a}$ that are only used internally to call the rules \drule{coind} and \drule{ind}.
	
	A type or multi-type is \emph{finite} whenever it can be formed
	without the rule \drule{coind}, and involves only finitely supported multi-types.
%
%	Recall that, by the assumptions we made above,
%	\begin{ienumerate}
%	\item the multisets involved are finitary, \ie
%		for all $b \in \Tropical \setminus \set \infty$
%		the set $\set{ a_i }[ b \geq a_i ]$ is finite,
%	\item all $\sigma_i$ are pairwise distinct.
%	\end{ienumerate}
\end{definition}

An example of a $\Tropical$-intersection type is the type $\sigma \eqdef [\gr{\sigma}{1}] \lto \alpha$ (rigorously, $\sigma$ is the greatest fixed point of this equation): it is well-formed because each coinductive call occurs inside a $\Tropical$-multiset, inside an element bearing grade $1$. Instead, $\sigma \eqdef [\gr{\sigma}{0}] \lto \alpha$ would not define a valid type in $\itypes$.

Typing judgments are of the shape $\Gamma \vdash M:\sigma$,
with auxiliary judgments of the shape $\Gamma \vdash M:\ms m$.
The context $\Gamma$ is a map $\Var \to \itypes!$
such that $\support(\Gamma) \eqdef \set{x \in \Var}[\Gamma(x) \neq []]$
is a finite set $\set{x_1,\dots,x_n}$.
We will present $\Gamma$ in the usual way, as a list
$x_1:\ms{m}_1, \dots, x_n:\ms{m}_n$.
The operations induced on $\msets[\Tropical]{\itypes}$
by the operations of $\Tropical$ are again mapped on contexts:
\[ \min(\Gamma, \Gamma') : x \mapsto \min(\Gamma(x), \Gamma'(x)) \qquad
	(a + \Gamma) : x \mapsto a + \Gamma(x). \]

\begin{definition}[Typing in $\itypes$] \label{def:itypes-typing}
	A λ-term $M$ is \emph{typable} in $\itypes$
	whenever there is a derivation $\derivD$
	with conclusion $\Gamma \vdash M : \sigma$
	in the following (inductive) system of typing rules:
	\begin{gather*}
		\begin{prooftree}
		\infer0[ax]{ \min(\Gamma, x:[\gr{\sigma}{0}]) \vdash x:\sigma }
		\end{prooftree}
	\qquad\qquad
		\begin{prooftree}
		\hypo{ \Gamma, x:\ms{m} \vdash M:\tau }
		\infer1[\lto_i]{ \Gamma \vdash λx.M : \ms{m} \lto \tau }
		\end{prooftree}
	\\[\topsep]
		\begin{prooftree}
		\hypo{ \Gamma \vdash M : \ms{m} \lto \tau }
		\hypo{ \Gamma' \vdash N : \ms{m} }
		\infer2[\lto_e]{ \min(\Gamma, \Gamma') \vdash MN : \tau }
		\end{prooftree}
	\qquad
		\begin{prooftree}
		\hypo{ \left[ \gr {(\Gamma_i \vdash N:\sigma_i)} {a_i} \right]
			_{i \in I} }
		\infer1[!]{ \min_{i \in I} (a_i + \Gamma_i)
			\vdash N : 
			[\gr{\sigma_i}{a_i}]_{i \in I}
			}
		\end{prooftree}
	\end{gather*}
	for some context $\Gamma$ and $\Tropical$-intersection type $\sigma$.
	We then write $\derivD \derives \Gamma \vdash M : \sigma$,
	and simply $\derives \Gamma \vdash M : \sigma$ 
	to say that such a $\derivD$ exists. 
\end{definition}

As an example, one can type Curry's fixed-point combinator in $\itypes$
as follows, 
by defining $\ms f \eqdef [\gr {[\gr{\alpha}{1}] \lto \alpha} 0]$,
$\ms m \eqdef [\gr {\ms m \lto \alpha} 1]$, 
and $\ms m' \eqdef [\gr {\ms m \lto \alpha} 0]$:
\begin{equation} \label{eq:itype-derivation-y}
	\begin{smallprooftree}[center]
	\infer0{ f:\ms f \vdash f:[\gr{\alpha}{1}] \lto \alpha }
	\infer0{ x:\ms m' \vdash x: \ms m \lto \alpha}
	\infer0{ x:\ms m' \vdash x: \ms m \lto \alpha}
	\infer1{ x:\ms m \vdash x:\ms m }
	\infer2{ x:\ms m' \vdash xx : \alpha }
	\infer1{ x:\ms m \vdash xx : [\gr{\alpha}{1}] }
	\infer2{ f:\ms f, x:\ms m \vdash f(xx) : \alpha }
	\infer1{ f:\ms f \vdash \abs x.f(xx): \ms m \lto \alpha }
	\hypo{\vdots}
	\infer1{ f:\ms f \vdash \abs x.f(xx): \ms m \lto \alpha }
	\infer1{ f:\ms f +1 \vdash \abs x.f(xx): \ms m}
	\infer2{ f:\ms f \vdash (\abs x.f(xx))(\abs x.f(xx)):\alpha }
	\infer1{ \vdash Y: [\gr {[\gr{\alpha}{1}] \lto \alpha} 0] \lto \alpha }
	\end{smallprooftree}
\end{equation}
(We omit the brackets around $\Tropical$-multisets of hypotheses for better readability,
and present the hypotheses in an arbitrary order.
The omitted part of right branch is a copy of the left branch.)
Observe that the derivation and the grades are similar to those we presented in $\fmutypes$ in \cref{eq:fmutypes-derivation-y}.

\begin{remark}
	The whole material we introduced so far can actually be built \wrt a generic preordered semiring $\Semiring$ (instead of $\Tropical$).
	In particular, usual idempotent (\resp non-idempotent) intersection types are obtained by taking $\Semiring$ to be the boolean semiring (\resp the semiring of naturals).
	This is part of an ongoing work programme; a preliminary presentation
	is given in \autocite{Cerda.Dal.26a}.
\end{remark}

Despite its simplicity, the type system $\itypes$ has all the necessary structure to characterise the hereditary head normalising $\lambda$-terms. The next two subsections are devoted to proving this correspondence. We first state the characterisation theorem in \cref{sec:it:sec:soundcomplete}, and then develop its proof, based on approximations and a metric argument, in \cref{sec:it:sec:complete-proof}.

%****************************************************************************
\subsection{Soundness and Completeness}
\label{sec:it:sec:soundcomplete}

The first step in the analysis of our type system is, as usual, subject reduction; in any respectable intersection type system, the dual property also holds, namely that types are preserved backwards.

\begin{theorem} \label{the:itypes-sr-se}
	The typing rules of $\itypes$ enjoy
	subject reduction and subject expansion:
	for all $M,N \in \pureterms$ such that $M \bred N$,
	$\derives \Gamma \vdash M : \sigma$
	\ifandonlyif $\derives \Gamma \vdash N : \sigma$.
\end{theorem}

The proof follows the usual forward and backward substitution arguments \autocite[see \eg,][Theorems~4.1 and~4.3]{Bucciarelli.Kes.Ven.17}.
The only subtlety, which by the way is also hidden in the proof for usual \emph{idempotent} intersection typing, arises from the identity $\min(0,0) = 0$: a unit contribution need not have a unique source. Concretely, from any derivation
\[	\begin{smallprooftree}
	\infer0{ x : [\gr \sigma 0] \vdash x : \sigma }
	\infer1{ \vdash λx.x : [\gr \sigma 0] \lto \sigma }
	\hypo{\vdots}
	\infer1{ \Delta \vdash N : \sigma }
	\hypo{\vdots}
	\infer1{ \Delta' \vdash N : \sigma }
	\infer2{ \min(\Delta,\Delta') \vdash N : [\gr \sigma 0]}
	\infer2{ \min(\Delta,\Delta') \vdash (λx.x)N : \sigma }
	\end{smallprooftree} \]
with $\Delta \neq \Delta'$, one can only retrieve a derivation of $\min(\Delta,\Delta') \vdash N : \sigma$ by using weakening, which is why we needed to define a non-relevant rule \drule{ax} in \cref{def:itypes-typing}.

What properties do typing in $\itypes$ capture?
The following Theorem gives a first answer to this question.

\begin{theorem} \label{the:itypes-sound-complete-for-hn}
	Typing in $\itypes$ is sound and complete 
	for head normalisation:
	for all $M \in \pureterms$, 
	$M \in \HN$ if and only if
	there are $\Gamma$ and $\sigma$ such that
	$\derives \Gamma \vdash M : \sigma$.
\end{theorem}

	\begin{proof}
	\emph{Soundness.}
	Recall the reducibility candidates (\rc) from \cref{def:rc}:
	just as in \cref{def:redinterp},
	given an environment $\rho$ mapping every atom in $\Atoms$ to a \rc,
	the reducibility interpretation $\redinterp \sigma$
	of a type $\sigma$ is the \rc defined inductively by:
	\begin{gather*}
		\redinterp{ \alpha } \eqdef \rho(\alpha), \qquad
		\redinterp{ \ms m \lto \sigma} \eqdef
			\redinterp{\ms m} \to \redinterp \sigma, \qquad
		\redinterp{\ms m} \eqdef \bigcap_{\substack{
			\sigma \in \itypes \\ \ms m(\sigma) = 0
			}} \redinterp{\sigma},
	\end{gather*}
	where the empty intersection is defined to be $\pureterms$.
	Then we prove the following statement,
	which is an analogue of \cref{lem:if-typable-then-in-rc}:
	if $\derives x_1:\ms m_1, \dots, x_n:\ms m_n \vdash M:\sigma$,
		$\rho$ is an environment, and
		for all $i \in [1,n]$, $N_i \in \redinterp{\ms m_i}$,
		then $M \subst[\vec x]{\vec N}
		\eqdef M \subst[x_1]{N_1} \cdots \subst[x_n]{N_n}
		\in \redinterp \sigma$.
	The proof is by induction on the first hypothesis.
	Now suppose $\derives \Gamma \vdash M:\sigma$.
	Apply the statement we just proved to this derivation,
	the environment $\rho_\HN$ mapping all atoms to $\HN$,
	and $N_i \eqdef x_i$ for all $x_i$ in the support of~$\Gamma$.
	We obtain that $M \in \redinterp[\rho_\HN]{\sigma} \subseteq \HN$
	by the definition of a \rc.
	\par\noindent
	\emph{Completeness.} A standard argument shows that any head normal form
	is typable \autocite[see \eg,][Lemma~5.5]{Bucciarelli.Kes.Ven.17}.
	By subject expansion (\cref{the:itypes-sr-se})
	all terms in $\HN$ are typable.
	\qed
	\end{proof}

The result above is neither novel nor surprising: indeed, the infinitary nature of our types and the tropical grade algebra play no role here. We have simply shown that the standard proof of soundness and completeness \wrt head normalisation for (finite) intersection types applies to $\itypes$. Our interest, however, lies in characterising a smaller class of terms, namely those that are not only head-normalising but \emph{hereditarily} so.

In the standard setting of non-idempotent ($\Nat$-graded) intersection types, there are (at least) two further results of interest when attention is restricted to types in which the empty multiset never occurs in positive position:
\begin{itemize}
 \item with finite types, the resulting system is sound and complete for
	 (weak) β-normalisation.
	 \autocites[Theorem~6.3.27]{deCarvalho.07}[§~6]{Bucciarelli.Kes.Ven.17};
 \item with fully infinite types
	 (in the sense that there is no restriction on the coinductive guards)
	 satisfying an additional \enquote{approximability} condition,
	 the resulting system can be made sound and complete for hereditary head normalisation
	 \autocite{Vial.17,Vial.21}.
 \end{itemize}
Our result below characterises terms in $\HHN$, as does the latter.
(As a subset of $\pureterms$, this set is defined by considering only the first three clauses of \cref{def:productive}.)
At the same time, it retains the cleaner flavour of the former, in that it applies to types whose infinite branches satisfy a guard condition, while requiring \emph{no side condition} beyond the restriction that occurrences of the empty multiset be negative (which could be directly internalised in the typing system, as the presentation of the following definition shows).
This is the main result of this Section.

\begin{definition}[Positive typing] \label{def:itypes-positive}
	The subsets $\itypes+ \subset \itypes$ and $\itypes- \subset \itypes$
	of types where the empty $\Tropical$-multiset does not occur
	in positive (\resp negative) position
	are mutually defined by refining the formation rules
	of \cref{def:itypes} as follows:
	\begin{gather*}
		\begin{prooftree}
		\hypo{ \alpha \in \Atoms }
		\infer1{ \alpha \in \itypes+ }
		\end{prooftree}
	\qquad
		\begin{prooftree}
		\hypo{ \ms m \in \itypes!- }
		\hypo{ \sigma \in \itypes+ }
		\infer2{ \ms m \lto \sigma \in \itypes+ }
		\end{prooftree}
	\qquad
		\begin{prooftree}
		\hypo{ \left[ \gr {(\sigma_i \inlater{a_i} \itypes+)} 
			{a_i} \right]_{i \in I} }
		\infer1{ [\gr{\sigma_i}{a_i}]_{i \in I} \in \itypes!+ }
		\end{prooftree}
	\\[\topsep]
		\begin{prooftree}
		\hypo{ \alpha \in \Atoms }
		\infer1{ \alpha \in \itypes- }
		\end{prooftree}
	\qquad
		\begin{prooftree}
		\hypo{ \ms m \in \itypes!+ \setminus \set{[]} }
		\hypo{ \sigma \in \itypes- }
		\infer2{ \ms m \lto \sigma \in \itypes- }
		\end{prooftree}
	\qquad
		\begin{prooftree}
		\hypo{ \left[ \gr {(\sigma_i \inlater{a_i} \itypes-)} 
			{a_i} \right]_{i \in I} }
		\infer1{ [\gr{\sigma_i}{a_i}]_{i \in I} \in \itypes!- }
		\end{prooftree}
	\end{gather*}
	where the rules for \drule{ind} and \drule{coind} remain unchanged.
	Then we write $\derivD \derives+ \Gamma \vdash M:\sigma$
	if there is a derivation $\derivD \derives \Gamma \vdash M:\sigma$
	with $\Gamma : \Var \to \itypes!-$ and $\sigma \in \itypes+$.
	\qed
\end{definition}

\begin{theorem} \label{the:pos-itypes-sound-complete-for-hhn}
	Positive typing in $\itypes$ is sound and complete 
	for hereditary head normalisation:
	for all $M \in \pureterms$, 
	$M \in \HHN$ \ifandonlyif
	there are $\Gamma$ and $\sigma$ such that
	$\derives+ \Gamma \vdash M : \sigma$.
\end{theorem}
	
	\begin{proof}[Proof of Soundness]
	Suppose $\derives+ \Gamma \vdash M : \sigma$.
	By \cref{the:itypes-sound-complete-for-hn},
	$M \in \HN$: there is a β-reduction
	$M \bred* λx_1.\dots λx_m.yM_1\dots M_n$.
	By subject reduction (\cref{the:itypes-sr-se}),
	there is a derivation
	$\derivD \derives+ \Gamma \vdash λx_1.\dots λx_m.yM_1\dots M_n : \sigma$.
	By the typing rules and by positivity
	(\cref{def:itypes-typing,def:itypes-positive}),
	the head occurrence of $y$ is given a type
	$\ms m_1 \lto \dots \lto \ms m_n \lto \tau$ in $\derivD$,
	for some $\ms m_1,\dots,\ms m_n \in \itypes!+ \setminus \set{[]}$
	and $\tau \in \itypes-$.
	Since all the $\ms m_j$ are non-empty, $\derivD$ has sub-derivations
	$\derives+ \Gamma_j \vdash M_j : \sigma_j$ for all $j \in [1,n]$.
	By coinduction, all $M_j$ are productive, hence so is $M$.
	\qed
	\end{proof}
The proof of completeness is the purpose of \cref{sec:it:sec:complete-proof},
and concludes with \cref{cor:pos-itypes-complete-for-hhn}.

%****************************************************************************
\subsection{Proof of Completeness}
\label{sec:it:sec:complete-proof}

We prove completeness by constructing, for every $M\in\HHN$,
a convergent sequence of finite typing derivations of $M$.
However, in order to define a suitable notion of convergence and to study how it preserves positive typing,
we need to work in an auxiliary system that we first introduce.
We will then provide a high-level summary of the proof, before diving in its technicalities.

\subsubsection{An Auxiliary System.}

The proof is carried out in an auxiliary system:
instead of taking single grades in $\Tropical$
we consider \emph{multisets of tropical grades},
\ie we assign a grade to \emph{all occurrences}
of the argument of a function,
instead of retaining only the minimal grade
(the minimal delay with which the argument is used).
For a fixed argument type, the grade $[0,0,2]$, for example,
records two uses at grade $0$ and one at grade $2$.
Taking its minimum would retain only $0$.
This additional information allows us to track occurrences
when transforming and comparing typing derivations.
The resulting system 
can be seen as a \enquote{non-idempotent} version of $\itypes$.

\begin{definition}[$\msets{\Tropical*}$-Multisets]
\label{def:semiring-of-msets}
	Define $\Tropical* \eqdef \Tropical \setminus \set\infty$.
	The set $\msets{\Tropical*}$ is defined to be $\Tropical* \to \Nat$.
	It is equipped with the following operations,
	forming a semiring structure
	(with respective neutral elements $[]$ and $[0]$):
	\begin{gather*}
		(\ms a \mstplus \ms b) : c \mapsto \ms a(c) + \ms b(c) \qquad
		(\ms a \msttimes \ms b) : c \mapsto \sum_{c = a + b}
			\ms a(a) \times \ms b(b).
	\end{gather*}
	As in \cref{def:i:tropical-msets}, 
	we introduce the set $\msets[\msets{\Tropical*}]{X}$ 
	of \emph{$\msets{\Tropical*}$-multisets} on the set $X$ 
	to be the set of all $\ms m : X \to \msets{\Tropical*}$
	such that for all $d \in \Tropical*$, 
	the set $\set{x\in X}[\min(\ms m(x)) = d]$ is finite.
	We extend to this notion all the notations of that Definition.
\end{definition}
Observe that the sum $\mstplus$ is just multiset union
(\ie, it is the pointwise mapping of the sum $+$ of $\Nat$);
the product $\msttimes$, instead, is given by convolution
using the product $+$ of $\Tropical$.
As in \cref{def:semiring-of-msets} we will denote 
the multisets in $\msets{\Tropical*}$, that we will use as grades,
by $\ms a$, $\ms b$, etc. in order to distinguish them
from $\msets{\Tropical*}$-multisets of types,
which we still denote by $\ms m$, $\ms n$, etc.

\begin{definition}[{$\msets{\Tropical*}$-Intersection Types and Typing}]
	As in \cref{def:itypes,def:itypes-typing},
	the sets $\itypes[\msets{\Tropical*}]$
	and $\itypes![\msets{\Tropical*}] \subset \msets[\msets{\Tropical*}]{\itypes[\msets{\Tropical*}]}$
	are defined by mutual induction and coinduction
	by the following system of formation rules:
	\begin{gather*}
		\begin{prooftree}
		\hypo{ α \in \Atoms }
		\infer1[\Atoms]{ α \in \itypes[\msets{\Tropical*}] }
		\end{prooftree}
	\qquad
		\begin{prooftree}
		\hypo{ \ms{m} \in \itypes![\msets{\Tropical*}] }
		\hypo{ \tau \in \itypes[\msets{\Tropical*}] }
		\infer2[\lto]{ \ms{m} \lto \tau \in \itypes[\msets{\Tropical*}] }
		\end{prooftree}
	\\[\topsep]
		\begin{prooftree}[center]
		\hypo{ \left[ \gr {(\sigma_i \inlater{\ms a_i} 
			\itypes[\msets{\Tropical*}])} {\ms a_i} \right]_{i \in I} }
		\infer1[!]
			{ [\gr{\sigma_i}{\ms a_i}]_{i \in I} \in \itypes![\msets{\Tropical*}] }
		\end{prooftree}
	\quad
		\begin{prooftree}[center]
		\hypo{ \sigma \in \itypes[\msets{\Tropical*}] }
		\hypo{ \ms a(0) = 0 }
		\infer[double]2[coind]
			{ \sigma \inlater{\ms a} \itypes[\msets{\Tropical*}] }
		\end{prooftree}
	\quad
		\begin{prooftree}[center]
		\hypo{ \sigma \in \itypes[\msets{\Tropical*}] }
		\infer1[ind]
			{ \sigma \inlater{\ms a} \itypes[\msets{\Tropical*}] }
		\end{prooftree}
	\end{gather*}
	and, for any finitely supported context
	$\Gamma : \Var \to \itypes![\msets{\Tropical*}]$,
	any term $M \in \pureterms$ and any type $\sigma \in \itypes[\msets{\Tropical*}]$,
	we say that $M$ is typable in $\itypes[\msets{\Tropical*}]$
	whenever the sequent $\Gamma \vdash M : \sigma$ is derivable
	in the following set of typing rules:
	\begin{gather*}
		\begin{prooftree}
		\infer0[ax]{ x:[\gr{\sigma}{[0]}] \vdash x:\sigma }
		\end{prooftree}
	\qquad
		\begin{prooftree}
		\hypo{ \Gamma, x:\ms{m} \vdash M:\tau }
		\infer1[\lto_i]{ \Gamma \vdash λx.M : \ms{m} \lto \tau }
		\end{prooftree}
	\\[\topsep]
		\begin{prooftree}
		\hypo{ \Gamma \vdash M : \ms{m} \lto \tau }
		\hypo{ \Gamma' \vdash N : \ms{m} }
		\infer2[\lto_e]{ \Gamma \mstplus \Gamma' \vdash MN : \tau }
		\end{prooftree}
	\qquad
		\begin{prooftree}
		\hypo{ \left[ \gr {(\Gamma_i \vdash N:\sigma_i)} {\ms a_i} \right]
			_{i \in I} }
		\infer1[!]{ \mstsum_{i \in I} \ms a_i \msttimes \Gamma_i
			\vdash N : 
			[\gr{\sigma_i}{\ms a_i}]_{i \in I}
			}
		\end{prooftree}
	\end{gather*}
\end{definition}

Observe that this time, as opposed to \cref{def:itypes-typing}, we chose a \emph{relevant} $\drule{ax}$ rule: indeed it eases the subsequent proofs, without breaking subject reduction and expansion (see \cref{the:itset-itypes-subjexp} for what we need in this paper, and \autocite{Cerda.Dal.26a} for a broader discussion).
The reason why we will prove completeness in this apparently more complicated
system is the following crucial observation:

\begin{remark} \label{obs:itin-maps-msets-to-tropicals}
	Replacing all grades $\ms a \in \msets{\Tropical*}$
	with $\min(\ms a) \in \Tropical$ turns
	formation and typing rules of $\itypes[\msets{\Tropical*}]$
	into formation and typing rules of $\itypes[\Tropical]$.
	This defines a map taking types, contexts and typing derivations
	from the former system to the latter.
\end{remark}

\subsubsection{The Proof at a Glance.}

Before giving the technical details of the completeness proof,
let us illustrate the construction
on our running example, Curry's fixed-point combinator $Y$.
We concentrate on the subterm
$Y_f\eqdef(\abs x.f(xx))(\abs x.f(xx))$,
whose head normalisation gives $Y_f\hred* f(Y_f)$.
We construct finite derivations of $Y_f$,
each exploring one more layer of head normalisation,
and obtain a positive derivation as their limit.

\smallskip\noindent
\emph{Finite approximations.}
Put
$\ms f' \eqdef [\gr{[\gr{\alpha}{[1]}]\lto\alpha}{[0]}]$ and
$\ms f \eqdef [\gr{[\gr{\alpha}{[1]}]\lto\alpha}{[0,1,\dots]}]$.
Our goal is a derivation of $f:\ms f\vdash Y_f:\alpha$.
For each $n\in\Nat$, we construct
$\derivD_n\derives f:\ms f_n\vdash Y_f:\alpha$, where
$
	\ms f_n
	\eqdef
	[\gr{[]\lto\alpha}{[n]},
	 \gr{[\gr{\alpha}{[1]}]\lto\alpha}{[0,\dots,n-1]}].
$
The first component corresponds to the point where the approximation stops;
the second records the applications of $f$ already explored.

We first type the head normal form $f(Y_f)$.
At stage $0$, we ignore its argument by giving it the empty multi-type.
At stage $n+1$, we use $\derivD_n$ at grade $[1]$:
\begin{align*}
	\derivD'_0
	&\quad\eqdef\quad
	\begin{smallprooftree}
		\infer0[ax]{
			f:\ms f_0\vdash f:[]\lto\alpha
		}
		\infer0[!]{
			\vdash Y_f:[]
		}
		\infer2[\lto_e]{
			f:\ms f_0\vdash f(Y_f):\alpha
		}
	\end{smallprooftree}
	\\[\topsep]
	\derivD'_{n+1}
	&\quad\eqdef\quad
	\begin{smallprooftree}
		\infer0[ax]{
			f:\ms f'\vdash f:
			[\gr{\alpha}{[1]}]\lto\alpha
		}
		\hypo{\derivD_n}
		\infer1{
			f:\ms f_n\vdash Y_f:\alpha
		}
		\infer1[!]{
			f:[1]\msttimes\ms f_n
			\vdash Y_f:[\gr{\alpha}{[1]}]
		}
		\infer2[\lto_e]{
			f:\ms f_{n+1}\vdash f(Y_f):\alpha
		}
	\end{smallprooftree}
\end{align*}
Indeed,
$\ms f_{n+1}=\ms f'\mstplus([1]\msttimes\ms f_n)$.

\smallskip\noindent
\emph{Subject expansion.}
We transport these derivations back to $Y_f$ by defining derivations
$
	\derivD_n
	\eqdef
	\subjexp{Y_f\hred* f(Y_f)}(\derivD'_n).
$
The subject-expansion map preserves finite derivations
and does not increase distances
(\cref{the:itset-itypes-subjexp}).
To make this step explicit, put
$\ms m_0\eqdef[]$ and
$\ms m_{n+1} \eqdef 
\left[ \gr{\ms m_i\lto\alpha}{[n+1-i]} \right]_{0\leq i\leq n}$.
The first two expanded derivations are:
\begin{align*}
	\derivD_0
	&\quad\eqdef\quad
	\begin{smallprooftree}
		\infer0[ax]{
			f:\ms f_0\vdash f:[]\lto\alpha
		}
		\infer0[!]{
			\vdash xx:[]
		}
		\infer2[\lto_e]{
			f:\ms f_0\vdash f(xx):\alpha
		}
		\infer1[\lto_i]{
			f:\ms f_0\vdash\abs x.f(xx):[]\lto\alpha
		}
		\infer0[!]{
			\vdash\abs x.f(xx):[]
		}
		\infer2[\lto_e]{
			f:\ms f_0\vdash Y_f:\alpha
		}
	\end{smallprooftree}
	\\[\topsep]
	\derivD_1
	&\quad\eqdef\quad
	\begin{smallprooftree}
		\infer0[ax]{
			f:\ms f'\vdash f:
			[\gr{\alpha}{[1]}]\lto\alpha
		}
		\infer0[ax]{
			x:[\gr{[]\lto\alpha}{[0]}]
			\vdash x:[]\lto\alpha
		}
		\infer0[!]{
			\vdash x:[]
		}
		\infer2[\lto_e]{
			x:[\gr{[]\lto\alpha}{[0]}]\vdash xx:\alpha
		}
		\infer1[!]{
			x:\ms m_1\vdash xx:[\gr{\alpha}{[1]}]
		}
		\infer2[\lto_e]{
			f:\ms f',x:\ms m_1\vdash f(xx):\alpha
		}
		\infer1[\lto_i]{
			f:\ms f'\vdash\abs x.f(xx):\ms m_1\lto\alpha
		}
		\hypo{\text{left branch of $\derivD_0$}}
		\infer1{
			f:\ms f_0\vdash\abs x.f(xx):[]\lto\alpha
		}
		\infer1[!]{
			f:[1]\msttimes\ms f_0
			\vdash\abs x.f(xx):\ms m_1
		}
		\infer2[\lto_e]{
			f:\ms f_1\vdash Y_f:\alpha
		}
	\end{smallprooftree}
\end{align*}

\smallskip\noindent
\emph{Convergence.}
The ultrametric defined below discounts differences at grade $a$
by a factor $2^{-a}$.
The grade $[1]$ used at each stage, together with
non-expansiveness of subject expansion, therefore gives
$\dist(\derivD_{p+1},\derivD_{q+1}) \leq \frac12\dist(\derivD_p,\derivD_q)$,
hence $\dist(\derivD_p,\derivD_q)\leq2^{-\min(p,q)}$.
Because the space of derivations is complete
(\cref{lem:dist-on-derivations}),
the sequence has a limit.

\smallskip\noindent
\emph{Positivity of the limit.}
Convergence alone does not ensure positivity.
In the conclusion of $\derivD_n$, the positive occurrence of $[]$
lies in the domain of $[]\lto\alpha$, at \enquote{temporal depth} $n$,
\ie depth measured by the sum of the grades traversed.
Every positive multi-type reached at a smaller depth is
$[\gr{\alpha}{[1]}]$, hence stays at distance $\frac12$ from $[]$.
This bound prevents a positive multi-type from becoming empty
in the limit; it is the key invariant in
\cref{lem:it:guarded-approximation}.

Concretely, $\lim_n\ms f_n = \ms f$ and
$\lim_n\ms m_n = \ms m \eqdef [\gr{\ms m\lto\alpha}{[1,2,\dots]}]$.
The limit derivation, followed by abstraction over $f$, is:
\[
	\begin{smallprooftree}[center]
		\infer0[ax]{
			f:\ms f'\vdash f:
			[\gr{\alpha}{[1]}]\lto\alpha
		}
		\infer0[ax]{
			x:[\gr{\ms m\lto\alpha}{[0]}]
			\vdash x:\ms m\lto\alpha
		}
		\infer0[ax]{
			x:[\gr{\ms m\lto\alpha}{[0]}]
			\vdash x:\ms m\lto\alpha
		}
		\infer1[!]{
			x:\ms m\vdash x:\ms m
		}
		\infer2[\lto_e]{
			x:[\gr{\ms m\lto\alpha}{[0,1,\dots]}]
			\vdash xx:\alpha
		}
		\infer1[!]{
			x:\ms m\vdash xx:[\gr{\alpha}{[1]}]
		}
		\infer2[\lto_e]{
			f:\ms f',x:\ms m\vdash f(xx):\alpha
		}
		\infer1[\lto_i]{
			f:\ms f'\vdash\abs x.f(xx):\ms m\lto\alpha
		}
		\hypo{\text{same derivation}}
		\infer1{
			f:\ms f'\vdash\abs x.f(xx):\ms m\lto\alpha
		}
		\infer1[!]{
			f:[1]\msttimes\ms f
			\vdash\abs x.f(xx):\ms m
		}
		\infer2[\lto_e]{
			f:\ms f\vdash Y_f:\alpha
		}
		\infer1[\lto_i]{
			\vdash Y:\ms f\lto\alpha
		}
	\end{smallprooftree}
\]
Taking the minimum of each grade multiset
(\cref{obs:itin-maps-msets-to-tropicals})
recovers the derivation in $\itypes$
from \cref{eq:itype-derivation-y}.

For an arbitrary term in $\HHN$, the proof follows the same construction:
we approximate each argument of its head normal form at grade $[1]$,
then apply subject expansion.
The convergence and positivity arguments are captured by
\cref{lem:it:guarded-approximation};
taking minima gives completeness for $\itypes$
(\cref{cor:pos-itypes-complete-for-hhn}).
The remaining subsubsections formalise this argument.
Readers may proceed directly to \cref{sec:it:sec:complexity}
in order to skip these technicalities.

\subsubsection{The Ultrametric Structure of Types and Derivations.}
\label{sec:itetric}

We equip the types of $\itypes[\msets{\Tropical*}]$
with a \enquote{syntactic} ultrametric distance,
\ie a distance such that every infinitary type
is the limit of a Cauchy sequence of finite types.
This was the way infinitary syntax was defined before
the generalisation of coinductive techniques
\autocite{Arnold.Niv.80,Dershowitz.Kap.Pla.91,Kennaway.Klo.Sle.Vri.95};
conversely such a distance is associated to any coinductive syntax \autocite{Barr.93}.

First, let us equip $\Tropical$ with the standard ultrametric distance:
$\dist(a,a) \eqdef 0$, and $\dist(a,b) \eqdef 2^{-\min(a,b)}$ if $a \neq b$.
Now, let us introduce the key construction we will rely on:
suppose an ultrametric space $(X,\dist)$ is given,
we equip $\msets[\msets{\Tropical*}]X$ with an ultrametric distance.
To do so:
\begin{enumerate}
\item Thanks to currying, we consider each
	$\ms m \in \msets[\msets{\Tropical*}]X \subset X \to (\Tropical* \to \Nat)$
	as an element of $(X \times \Tropical*) \to \Nat$,
	\ie as the multiset
	$\mscur m \eqdef \mset{ (x,a) }[ x \in X,\ a \in \ms m(x) ]$.
\item Then, to define a distance between such multisets,
	we resort to a construction from optimal transportation \autocite{Villani.09}:
	the distance between $\ms m$ and $\ms n$ is the total cost
	of an optimal transportation plan of $\mscur m$ towards $\mscur n$.
	Concretely, pairs $(x,a) \in \mscur m$ are \emph{matched}
	with pairs $(y,b) \in \mscur n$;
	some pairs from both sides may remain unmatched,
	giving rise do \emph{deletions} from $\mscur m$
	and \emph{insertions} into $\mscur n$;
	a \emph{cost} is associated to matching, deletion, insertion
	(by means of the distances on $\Tropical$ and $X$);
	the total cost of this transportation plan is
	the maximum of the costs of each individual operation
	(the fact that we take the maximum here instead of, \eg, the sum,
	is what will ensure an ultrametric behaviour);
	finally, the distance $\dist(\ms m,\ms n)$ is the infimum
	of this total cost over all possible transportation plans.
	This construction is the \emph{$\infty$-Wasserstein distance},
	or \enquote{bottleneck} distance
	\autocites{Champion.DeP.Juu.08}[Chap.~3]{Santambrogio.15}.
\end{enumerate}
Formally it is defined as follows.

\begin{definition} \label{def:bottleneck-dist}
	Given $\ms m, \ms n \in \msets[\msets{\Tropical*}]X$,
	a \emph{pairing} (or \emph{transportation plan})
	between $\ms m$ and $\ms n$ 
	is a partial bijective multiset map between $\mscur m$ and $\mscur n$,
	\ie a multiset $\pi \in \msets{(X \times \Tropical*)^2}$ such that:
	\begin{gather*}
		\forall (x,a) \in X \times \Tropical*,\ 
		\deletions\pi(x,a) \eqdef
		\mscur m(x,a) - \sum_{(y,b) \in X \times \Tropical*} \pi((x,a),(y,b))
		\geq 0, \\
		\forall (y,b) \in X \times \Tropical*,\ 
		\insertions\pi(y,b) \eqdef
		\mscur n(y,b) - \sum_{(x,a) \in X \times \Tropical*} \pi((x,a),(y,b))
		\geq 0.
	\end{gather*}
	The set of all pairings between $\ms m$ and $\ms n$
	is denoted by $\pairings{\ms m}{\ms n}$.
	The \emph{deletions} $\deletions{\pi}$
	and \emph{insertions} $\insertions{\pi}$ of such a pairing $\pi$
	are the multisets of elements of $X \times \Tropical$
	defined by the expressions above.
	The pairing $\pi$ is extended to a total multiset map
	$\totpairing{\pi} \in \msets{((X \cup \set\bullet) \times \Tropical)^2}$
	by defining:
	\[	\totpairing{\pi} \eqdef \pi 
		+ \mset{((x,a),(\bullet,\infty))}[(x,a) \in \deletions{\pi}]
		+ \mset{((\bullet,\infty),(y,b))}[(y,b) \in \insertions{\pi}]. \]
	The map $\dist : \msets[\msets{\Tropical*}]X \times \msets[\msets{\Tropical*}]X \to [0,1]$ is defined by
	\[	\dist(\ms m,\ms n) \eqdef \inf_{\pi \in \pairings{\ms m}{\ms n}}
		\max_{p \in \totpairing\pi} \cost(p) \]
	where the \emph{matching cost} is defined by:
	\[	\cost((x,a),(y,b)) \eqdef \max(\dist(a,b), 2^{-a} \dist(x,y))
		= \begin{cases}
		\dist(a,b) & \text{if $a \neq b$,} \\
		2^{-a} \dist(x,y) & \text{otherwise.}
		\end{cases} \]
	In particular
	the \emph{deletion} and \emph{insertion costs} are, respectively, $\cost((x,a),(\bullet,\infty)) = 2^{-a}$ and
	$\cost((\bullet,\infty),(y,b)) = 2^{-b}$.
\end{definition}

We now use this construction to mutually define distances on
$\itypes[\msets{\Tropical*}]$ and on $\itypes![\msets{\Tropical*}]$:
the former makes recursive calls to the latter, whereas
the latter is defined to be the distance on
$\msets[\msets{\Tropical*}]{\itypes[\msets{\Tropical*}]}$
from \cref{def:bottleneck-dist},
which internally uses the former, again recursively.
This mutual definition is well-founded at each finite precision,
hence converging to a real number in $[0,1]$:
recursive calls at grade $0$ follow the inductive formation of types,
whereas recursive calls at a positive grade $a$ are scaled by $2^{-a}$
and therefore require strictly less precision.

\begin{definition} \label{def:dist-on-itypes}
	A map $\dist : \itypes[\msets{\Tropical*}] \times \itypes[\msets{\Tropical*}] \to [0,1]$ is defined by:
	\[	\dist(\alpha,\alpha) \eqdef 0 \qquad
		\dist(\ms m \lto \sigma, \ms n \lto \tau) \eqdef
			\max(\dist(\ms m, \ms n), \dist(\sigma, \tau)) \]
	where $\dist(\ms m, \ms n)$ is defined according to
	\cref{def:bottleneck-dist},
	and $\dist(\sigma,\tau) \eqdef 1$ otherwise.
\end{definition}

\begin{restatable}{lemma}{sysIdistisultrametric} 
\label{lem:dist-is-ultrametric}
	\Cref{def:bottleneck-dist,def:dist-on-itypes} equip
	$\itypes[\msets{\Tropical*}]$ and $\itypes![\msets{\Tropical*}]$
	with a complete ultrametric structure
	such that each type (\resp multi-type)
	is the limit of a Cauchy sequence of finite types (\resp multi-types).
\end{restatable}
The proof of this lemma is detailed in \cref{app:i:ssec:ultrametric}.

Finally we reach our goal, namely to define a distance on typing derivations
so that any derivation in $\itypes[\msets{\Tropical*}]$
can be approximated by derivations using only finite types:

\begin{lemma} \label{lem:dist-on-derivations}
	The set of all typing derivations in $\itypes[\msets{\Tropical*}]$
	is equipped with a complete ultrametric structure
	by the distance $\dist$ defined by:
	\begin{align*}
	\dist\left(
		\begin{smallprooftree}[center]
		\hypo{\strut}
		\infer1{ x:[\gr{\sigma}{[0]}] \vdash x:\sigma }
		\end{smallprooftree}
	,
		\begin{smallprooftree}[center]
		\hypo{\strut}
		\infer1{ x:[\gr{\tau}{[0]}] \vdash x:\tau }
		\end{smallprooftree}
	\right) & \eqdef \dist(\sigma, \tau), \\[\topsep]
	\dist\left(
		\begin{smallprooftree}[center]
		\hypo{\derivD}
		\infer1{ \Gamma \vdash λx.M:\ms m \lto \sigma }
		\end{smallprooftree}
	,
		\begin{smallprooftree}[center]
		\hypo{\derivD'}
		\infer1{ \Gamma' \vdash λx.M:\ms m' \lto \sigma' }
		\end{smallprooftree}
	\right) & \eqdef \dist(\derivD, \derivD'), \\[\topsep]
	\dist\left(
		\begin{smallprooftree}[center]
		\hypo{\derivD}
		\hypo{\ms E}
		\infer2{ \Gamma \vdash MN:\sigma }
		\end{smallprooftree}
	,
		\begin{smallprooftree}[center]
		\hypo{\derivD'}
		\hypo{\ms E'}
		\infer2{ \Gamma' \vdash MN:\sigma' }
		\end{smallprooftree}
	\right) & \eqdef \max(\dist(\derivD,\derivD'), \dist(\ms E,\ms E'))
	\end{align*}
	where $\dist(\ms E,\ms E')$ is defined as in
	\cref{def:bottleneck-dist}, and in all other cases
	$\dist(\derivD, \derivD') \eqdef 1$.
	In~addition, each derivation
	is the limit of a Cauchy sequence of derivations
	using only finite types and multi-types.
\end{lemma}

Observe that two typing derivations are at distance $1$ whenever they
type different terms.
Moreover, the maps taking a derivation to its conclusion
type and to its context can be proved to be non-expansive.
%one can show that
%for all derivations $\derivD \derives \Gamma \vdash M:\sigma$
%and $\derivD' \derives \Gamma' \vdash M:\sigma'$
%we have $\dist(\sigma,\sigma') \leq \dist(\derivD,\derivD')$
%as well as $\dist(\Gamma(x),\Gamma'(x)) \leq \dist(\derivD,\derivD')$
%for all $x \in \Var$.

\subsubsection{Completeness by Guarded Approximation.}

For the typing system $\itypes$ to be sound,
we needed to exclude positive occurrences of the empty
$\Tropical$-multiset from typing judgments
(\cref{def:itypes-positive,the:pos-itypes-sound-complete-for-hhn});
the same restriction is needed for $\itypes[\msets{\Tropical*}]$.
Completeness must therefore be proved for this restricted system.

\begin{definition}[Positive Typing]
\label{def:it:aux-positive}
	The positive and negative subsets of
	$\itypes[\msets{\Tropical*}]$ and
	$\itypes![\msets{\Tropical*}]$
	are defined as in \cref{def:itypes-positive},
	using the formation rules of the auxiliary system.
	We write $\derivD\derives+\Gamma\vdash M:\sigma$
	whenever $\derivD\derives\Gamma\vdash M:\sigma$ with
	$\sigma\in\itypes[\msets{\Tropical*}]+$ and
	$\Gamma:\Var\to\itypes![\msets{\Tropical*}]-$.
\end{definition}

In order to prove completeness, we need two more key lemmas;
the first one is subject expansion.
We refine the usual statement by making it
\begin{ienumerate}
\item \emph{proof-relevant},
	in the sense that a fixed reduction $M\bred* M'$
	determines a rewriting of typing derivations of $M'$
	into typing derivations of $M$,
\item \emph{non-expansive},
	by proving that this rewriting does not
	increase the distance between derivations.
\end{ienumerate}
The proof, based on an analysis of the usual backward substitution lemma, is given in \cref{app:i:subject-expansion}.

\begin{restatable}[Subject Expansion]{lemma}{sysIsubjectexpansion}
\label{the:itset-itypes-subjexp}
	Every fixed reduction $M\bred* M'$
	determines a map $\subjexp{M\bred* M'}$
	taking any derivation
	$\derivD'\derives\Gamma\vdash M':\sigma$
	to a derivation
	\[
		\subjexp{M\bred* M'}(\derivD')
		\derives\Gamma\vdash M:\sigma.
	\]
	This map takes finite derivations to finite derivations
	and is non-expansive:
	for all derivations $\derivD'$ and $\derivE'$ of $M'$,
	$
		\dist\bigl(
			\subjexp{M\bred* M'}(\derivD'),
			\subjexp{M\bred* M'}(\derivE')
		\bigr)
		\leq\dist(\derivD',\derivE').
	$
	\end{restatable}

The last missing piece is the following lemma isolating the convergence argument:
it applies to any family of maps having the above properties
of subject expansion.

\begin{restatable}[Guarded Approximation]{lemma}{sysIguardedapprox}
\label{lem:it:guarded-approximation}
	Fix an atom $\alpha$ and a set $\mathcal C\subseteq\pureterms$.
	For every $M\in\mathcal C$, fix a reduction
	$
		M\hred* H_M
		\eqdef\lambda x_1.\dots\lambda x_k.yM_1\dots M_l
	$ with $M_1, \dots, M_l \in \mathcal C$.
	Suppose that, for every $M\in\mathcal C$, a non-expansive map $F_M$
	takes derivations of $\Gamma\vdash H_M:\sigma$
	to derivations of $\Gamma\vdash M:\sigma$,
	and takes finite derivations to finite derivations.

	For every $M\in\mathcal C$, define derivations
	$\derivD'_n(M)$ of $H_M$ and
	$\derivD_n(M)\eqdef F_M(\derivD'_n(M))$
	by induction on $n$:
	\begin{itemize}
	\item construct $\derivD'_0(M)$
		by typing the head occurrence of $y$ at type
		$[]\lto\dots\lto[]\lto\alpha$, with $l$ empty domains,
		typing each $M_j$ at $[]$ by the nullary instance of rule \drule{!},
		then applying \drule{\lto_e} and \drule{\lto_i};
	\item at stage $n+1$, use
		$\derivD_n(M_j)\derives\Delta_j\vdash M_j:\tau_j$
		and the unary instance of rule \drule{!} to derive
		$[1]\msttimes\Delta_j \vdash M_j:[\gr{\tau_j}{[1]}]$
		for all $1\leq j\leq l$;
		construct $\derivD'_{n+1}(M)$
		by typing the head occurrence of $y$ at type
		$[\gr{\tau_1}{[1]}]\lto\dots\lto[\gr{\tau_l}{[1]}]\lto\alpha$,
		then applying \drule{\lto_e} and \drule{\lto_i}.
	\end{itemize}
	Writing
	$\derivD_n(M)\derives\Gamma_n(M)\vdash M:\sigma_n(M)$,
	these derivations are finite and form a converging sequence such that
	$
		\lim_n\derivD_n(M)
		\derives+
		\lim_n\Gamma_n(M)
		\vdash M:\lim_n\sigma_n(M).
	$
\end{restatable}
The proof is given in \cref{app:i:guarded-approx}.
The main difficulty is that convergence does not preserve positivity,
even when every approximation is positive.
For example, put
$\tau_n\eqdef[\gr{\alpha}{[n]}]\lto\alpha$.
The positive axiom derivations
$x:[\gr{\tau_n}{[0]}]\vdash x:\tau_n$
converge to the axiom typing
$x:[\gr{\tau}{[0]}] \vdash x : \tau\eqdef[]\lto\alpha$.
This limit is not positive:
the empty domain of $\tau$ occurs positively in the context.
The non-empty domains of the approximations have disappeared
because their grades tend to infinity.

The key argument is therefore to establish a stronger invariant
for the constructed approximations:
every positive multi-type reached at \enquote{temporal depth} less than $n$
contains an element at grade at most $1$,
where \enquote{temporal depth} is measured by the sum of the grades traversed.
Such a multi-type stays at distance at least $\frac12$ from $[]$.
Moreover, any finite descent through a limit type
can be followed in sufficiently late approximations
with the same grades.
This prevents the vanishing of the content of successive multisets that we illustrated above.
\smallskip

This leads us to the concluding result of this subsection;
the difficult part of the main theorem of the paper,
namely \cref{the:pos-itypes-sound-complete-for-hhn},
follows as a straightforward corollary.

\begin{theorem}
\label{the:pos-mset-itypes-sound-complete-for-hhn}
	Positive typing in $\itypes[\msets{\Tropical*}]$
	is sound and complete for hereditary head normalisation:
	for all $M\in\pureterms$,
	$M\in\HHN$ \ifandonlyif
	there are $\Gamma$ and $\sigma$ such that
	$\derives+\Gamma\vdash M:\sigma$.
\end{theorem}

	\begin{proof}
	Soundness is immediate from the soundness part of
	\cref{the:pos-itypes-sound-complete-for-hhn},
	using \cref{obs:itin-maps-msets-to-tropicals},
	since taking minima preserves positivity.

	For completeness, fix a head reduction $M\hred* H_M$
	for every $M\in\HHN$, and take
	$F_M\eqdef\subjexp{M\hred* H_M}$.
	By \cref{the:itset-itypes-subjexp},
	these maps preserve the conclusion context and type,
	take finite derivations to finite derivations,
	and are non-expansive.
	\Cref{lem:it:guarded-approximation} for $\mathcal C \eqdef \HHN$ 
	therefore gives a positive limit derivation for every $M\in\HHN$.
	\qed
	\end{proof}

\begin{corollary}%
[Second part of \cref{the:pos-itypes-sound-complete-for-hhn}]
\label{cor:pos-itypes-complete-for-hhn}
	Positive typing in $\itypes$ is complete
	for hereditary head normalisation.
\end{corollary}

	\begin{proof}
	Immediate from the completeness part of
	\cref{the:pos-mset-itypes-sound-complete-for-hhn},
	thanks to \cref{obs:itin-maps-msets-to-tropicals}.
	\qed
	\end{proof}

%****************************************************************************
\subsection{On the Recursion-Theoretic Optimality of the Typing System}
\label{sec:it:sec:complexity}

To conclude, we investigate the optimality of our typing system.
It is well-known that $\HN$ is not decidable,
so that (by \cref{the:itypes-sound-complete-for-hn})
typability in unrestricted intersection type systems
is undecidable in general.
Therefore our point is not to investigate the complexity
of an effective typing procedure for $\itypes$,
but to resort to recursion theory to establish that 
typability in our typing system can be naturally described to be within the
same level of the arithmetical hierarchy than the one of
the class of λ-terms it characterises, namely $\HHN$.

As proved by \textcite{Tatsuta.08}, $\HHN$ is not recursively enumerable;
in fact, one can easily extend this result as follows.
Recall that a predicate $P(-)$ on $\pureterms$ is $\Pi_0^2$ whenever
\[	P(M) \iff \forall x \in \Nat,\ \exists y \in \Nat,\ \phi(\enc M, x,y) \]
for a given computable predicate $\phi$,
and where $\enc M$ denotes the encoding of $M$ as an integer.
A predicate is $\Pi_0^2$-complete if, in addition,
any other $\Pi_0^2$ predicate can be reduced to it.
A subset $X \subseteq \pureterms$ is $\Pi_0^2$(-complete)
if the associated predicate is so.

\begin{fact}
	The subset $\HHN$ of $\pureterms$ is $\Pi_0^2$-complete.
\end{fact}

Our claim that typability in $\itypes$ is an optimal characterisation
of $\HHN$ is formalised in the following result:

\begin{theorem}
	The subset $\mathcal T \eqdef \set{ M \in \pureterms}
	[ \exists \sigma,\ \exists \Gamma,\ 
	\derives+ \Gamma \vdash M : \sigma ]$ of $\pureterms$ is $\Pi_0^2$.
\end{theorem}
	
	\begin{proof}
	For every term $M$ and every $n\in\Nat$, attempt to construct
	the finite approximation $\derivD_n(M)$ as in
	\cref{lem:it:guarded-approximation}.
	This is an effective procedure:
	head reduction and the construction of the required argument
	derivations are followed by applications of typing rules
	and subject expansion; only finite objects are involved.
	Define $\phi(\enc M,n,t)$ to hold whenever
	the construction of $\derivD_n(M)$ terminates within $t$ steps.
	This predicate is computable.
	We claim that $M\in\mathcal T$ \ifandonlyif
	$\forall n\in\Nat,\ \exists t\in\Nat,\ \phi(\enc M,n,t)$,
	hence $\mathcal T$ is $\Pi_0^2$.
	
	For the forward implication, head normalisation and subject
	reduction give a positive typing of the head normal form of $M$.
	Each of its head arguments also admits a positive typing,
	as established in the proof of
	\cref{the:pos-itypes-sound-complete-for-hhn}.
	Termination of the construction follows by induction on $n$.
	Conversely, apply \cref{lem:it:guarded-approximation} with
	\[
		\mathcal C
		=\set{M\in\pureterms}[
			\forall n\in\Nat,\
			\text{the construction of $\derivD_n(M)$ terminates}
		].
	\]
	The resulting approximations converge to a positive derivation
	in the auxiliary system.
	Taking minima gives the required typing in $\itypes$.
	\qed
	\end{proof}

\section{Conclusion}
\label{sec:conclusion}
% !TeX root = ../main.tex
% !TeX spellcheck = en_GB

In this work, we have investigated the tropical semiring as the underlying grade algebra in a type-theoretic setting, showing how its temporal interpretation gives rise to a common guard mechanism for two complementary results. With $\fmutypes$, this mechanism \emph{guarantees} productive computation; with $\itypes$, it yields a \emph{characterisation} of hereditary head normalisation. We thus obtain, within the well-established framework of graded coeffects, a principled discipline for controlling infinite type structure and allowing normalisation guarantees to carry over to the setting of productivity.

\subsubsection{Further Work.}

A first general direction for future investigations towards more concrete applications of $\fmutypes$ is to refine this discipline by enriching the temporal requirements themselves. One could, \eg, introduce grade parameters and dependencies that describe observation depths or input lookahead determined during computation. Bounded linear logic, linear dependent types, and generalisations of resource-sensitive typing provide relevant precedents \autocite{Girard.Sce.Sco.92,DalLago.Gab.12,DalLago.Pet.14, Fukihara.Kat.21,Baillot.Sannier.26}. The main challenge here is to ensure a satisfactory interaction between grade variables and recursive types. Infinite trees and exact real computation would provide two concrete test cases for these extensions, bringing branching observations and precision-dependent input requirements into the study of tropical guards.

As for $\itypes$, we view this tropical instance as a first step towards a broader study of graded intersection types $\itypes[\Semiring]$. The questions suggested by our results concern which grading semirings $\Semiring$ support guarded infinite types, which operational properties the resulting systems characterise, and which algebraic assumptions are needed for their metatheory.

Understanding this broader picture also calls for a semantic account of the systems, which is the subject of ongoing work. Relational models, possibly arising from graded comonads \autocite{Carraro.Ehr.Sal.10} or enriched with weights \autocite{Laird.Man.McC.Pag.13}, are a natural direction to investigate, as they may clarify how temporal requirements interact with infinite type structure and with the graded coeffect operations.

Our completeness proof already suggests a concrete extension towards quantitative execution costs. If we replace tropical grades with multisets of tropical grades---the structure used in that proof---we retain both \emph{when} an argument is required and \emph{how many times} it is used at each stage. In the spirit of intersection types, this additional information should make it possible to extract fine-grained quantitative information from typing derivations. Following de Carvalho's analysis of Krivine-machine execution \autocite{deCarvalho.07,deCarvalho.17} and subsequent work on tight typings \autocite{Accattoli.Gra.Kes.20}, one could seek exact cost characterisations for computing each finite observation.

\begin{credits}
	\subsubsection{\ackname} 
	The authors wish to thank Tito Nguyễn and Gabriele Vanoni, who helped filling a gap in a preliminary version of this article.
%	
%	\subsubsection{\discintname}
%	The authors have no competing interests to declare that are relevant to the content of this article.
\end{credits}

\printbibliography

\clearpage
\appendix

\section[Appendix: Proofs about G\_T]{Proofs about $\fmutypes$ (\cref{sec:g})}
\label[appendix]{app:g}
% !TeX root = ../main.tex
% !TeX spellcheck = en_GB

\subsection{Subject Reduction of $\fmutypes$}
\label[appendix]{app:g:subject-reduction}

\linkedrestate{lem:g:sr}{\sysGsubjectreduction}

The proof essentially follows the standard argument for System~\textsf{F}.
We recall the main lemmas and emphasise the few differences.

\begin{lemma}[Context monotonicity]
\label{lem:app:g:context-monotonicity}
	Let $\Gamma, \Gamma'$ be two typing contexts such that 
	$\forall x \in \Var,$ $\pi_0(\Gamma'(x)) \le \pi_0(\Gamma(x))$
	and $\pi_1(\Gamma'(x)) = \pi_1(\Gamma(x))$,
	\ie, $\Gamma'$ assigns the same type as $\Gamma$ to~$x$,
	with a smaller grade.
	If $\Gamma\vdash M:A$, then $\Gamma'\vdash M:A$.
\end{lemma}

	\begin{proof}
	Immediate induction.
	\qed
	\end{proof}

\begin{lemma}[Term substitution]
	\label{lem:app:g:term-substitution}
	Suppose that $\Gamma,x:\bang[a]A\vdash M:B$ and $\Delta\vdash N:A$,
	where $\Gamma$ and $\Delta$ coincide on their second component (they assign the same types, possibly with different grades, to each variable), and $x$ is fresh for $\Gamma$, $\Delta$, and $N$.
	Then $\min(\Gamma,\Delta+a) \vdash M\subst{N}:B$.
\end{lemma}
	
	\begin{proof}
	By induction over the first given derivation, using \cref{lem:g:weakening,lem:app:g:context-monotonicity}.
	\qed
	\end{proof}

Then one defines type substitution on types as usual, and extends the definition to typing contexts (type substitution acts on assigned types without modifying the grades) and to terms (type substitution acts on annotations and argument types, recursively in subterms).

\begin{lemma}[Type substitution]
\label{lem:app:g:type-substitution}
	If $\Gamma\vdash M:A$ and $C\in\fmutypes$, then
	\[
	  \Gamma\tsubst[\alpha]{C}
	  \vdash
	  M\subst[\alpha]{C}
	  :
	  A\tsubst[\alpha]{C}.
	\]
\end{lemma}
	
	\begin{proof}
	The proof is the usual induction over the given derivation.
	One just has to check along the way that type substitution preserves guarded type formation: after renaming recursive binders fresh for the substituted type, substitution introduces no new occurrence of a bound recursive atom and leaves the guards on its existing occurrences unchanged.
	\qed
	\end{proof}

\begin{proof}[Proof of \cref{lem:g:sr}]
	If $M \bred M'$ occurs at top-level, we treat the three possible cases from \cref{def:bred}:
	term β-reduction is handled by \cref{lem:app:g:term-substitution},
	type β-reduction is handled by \cref{lem:app:g:type-substitution}, and
	the case of $\unfold/\fold$ elimination is immediate as
		the corresponding typing rules are inverse of each other.
	Otherwise, the result is lifted by an immediate induction on $M \bred M'$.
	\qed
\end{proof}

\subsection{Embedding Guarded Recursion Operators}
\label[appendix]{app:g:ssec:later-examples}

The combinators $\mathtt{next}_A$, $\mathtt{gfix}_A$ and $\mathtt{gapp}_{A,B}$ from \cref{sec:g:sec:later} are typed by the following derivations.

\begin{itemize}
\item For $a\in\{0,1\}$, the typing of $\termstyle{next}_A$
	is derived as follows:
	\begin{equation}
	\label{eq:app:g:next}
	\begin{smallprooftree}[center]
	  \infer0[ax]{
	    x:\bang[a]A,\ c:\bang[0](\ito[1]{A}{\bot})
	    \vdash c:\ito[1]{A}{\bot}
	  }
	  \infer0[ax]{
	    x:\bang[0]A,\ c:\bang[\infty](\ito[1]{A}{\bot})
	    \vdash x:A
	  }
	  \infer2[\lto_e]{
	    x:\bang[a]A,\ c:\bang[0](\ito[1]{A}{\bot})
	    \vdash cx:\bot
	  }
	  \infer1[\lto_i]{
	    x:\bang[a]A
	    \vdash\abs c[\ito[1]{A}{\bot}].cx:\lm A
	  }
	  \infer1[\lto_i]{
	    \vdash\termstyle{next}_A:\ito[a]{A}{\lm A}
	  }
	\end{smallprooftree}
	\end{equation}
	The application step uses $\min(a,1)=a$ and
	$\min(0,\infty+1)=0$.
\item For $\termstyle{gfix}_A$, write $Y$ for its type instance
	at $A$, of type $(\ito[1]{A}{A})\to A$.
	The complete derivation is given in \cref{fig:app:g:gfix-gapp} below.
\item For $\termstyle{gapp}_{A,B}$, abbreviate only its ambient context:
	\[
	  \Gamma\eqdef
	  f:\bang[0]\lm(A\to B),\
	  u:\bang[0]\lm A,\
	  k:\bang[0](\ito[1]{B}{\bot}).
	\]
	Since $\min(\Gamma,\Gamma+a)=\Gamma$ for every
	$a\in\Tropical$, this context is retained throughout the applications in the derivation, given in \cref{fig:app:g:gfix-gapp} below.
\end{itemize}

\begin{figure}[h!]
\centering
\rotatebox{90}{%
\scalebox{.65}{%
	\begin{prooftree}
	  \hypo{ \text{\cref{eq:fmutypes-derivation-y}} }
	  \infer1{
	    \vdash Y:(\ito[1]{A}{A})\to A
	  }
	  \infer1[weak]{
	    f:\bang[0](\lm A\to A)
	    \vdash Y:(\ito[1]{A}{A})\to A
	  }
	  \infer0[ax]{
	    f:\bang[0](\lm A\to A),\ x:\bang[1]A
	    \vdash f:\lm A\to A
	  }
	  \hypo{ \text{\cref{eq:app:g:next}} }
	  \infer1{
	    \vdash\termstyle{next}_A:\ito[1]{A}{\lm A}
	  }
	  \infer1[weak]{
	    f:\bang[0](\lm A\to A),\ x:\bang[1]A
	    \vdash\termstyle{next}_A:\ito[1]{A}{\lm A}
	  }
	  \infer0[ax]{
	    f:\bang[0](\lm A\to A),\ x:\bang[0]A
	    \vdash x:A
	  }
	  \infer2[\lto_e]{
	    f:\bang[0](\lm A\to A),\ x:\bang[1]A
	    \vdash\termstyle{next}_A x:\lm A
	  }
	  \infer2[\lto_e]{
	    f:\bang[0](\lm A\to A),\ x:\bang[1]A
	    \vdash f(\termstyle{next}_A x):A
	  }
	  \infer1[\lto_i]{
	    f:\bang[0](\lm A\to A)
	    \vdash
	    \abs x[A].f(\termstyle{next}_A x):
	    \ito[1]{A}{A}
	  }
	  \infer2[\lto_e]{
	    f:\bang[0](\lm A\to A)
	    \vdash
	    Y(\abs x[A].f(\termstyle{next}_A x)):A
	  }
	  \infer1[\lto_i]{
	    \vdash\termstyle{gfix}_A:(\lm A\to A)\to A
	  }
	\end{prooftree}%
}%
}
\hspace{3cm}
\rotatebox{90}{%
\scalebox{.65}{%
	\begin{prooftree}
	  \infer0[ax]{
	    \Gamma\vdash f:\lm(A\to B)
	  }
	  \infer0[ax]{
	    \Gamma,\ g:\bang[\infty](A\to B)
	    \vdash u:\lm A
	  }
	  \infer0[ax]{
	    \Gamma,\ g:\bang[\infty](A\to B),\
	    x:\bang[\infty]A
	    \vdash k:\ito[1]{B}{\bot}
	  }
	  \infer0[ax]{
	    \Gamma,\ g:\bang[0](A\to B),\
	    x:\bang[\infty]A
	    \vdash g:A\to B
	  }
	  \infer0[ax]{
	    \Gamma,\ g:\bang[\infty](A\to B),\
	    x:\bang[0]A
	    \vdash x:A
	  }
	  \infer2[\lto_e]{
	    \Gamma,\ g:\bang[0](A\to B),\
	    x:\bang[0]A
	    \vdash gx:B
	  }
	  \infer2[\lto_e]{
	    \Gamma,\ g:\bang[1](A\to B),\
	    x:\bang[1]A
	    \vdash k(gx):\bot
	  }
	  \infer1[\lto_i]{
	    \Gamma,\ g:\bang[1](A\to B)
	    \vdash\abs x[A].k(gx):\ito[1]{A}{\bot}
	  }
	  \infer2[\lto_e]{
	    \Gamma,\ g:\bang[1](A\to B)
	    \vdash u(\abs x[A].k(gx)):\bot
	  }
	  \infer1[\lto_i]{
	    \Gamma
	    \vdash
	    \abs g[A\to B].u(\abs x[A].k(gx)):
	    \ito[1]{(A\to B)}{\bot}
	  }
	  \infer2[\lto_e]{
	    \Gamma
	    \vdash
	    f(\abs g[A\to B].u(\abs x[A].k(gx))):\bot
	  }
	  \infer1[\lto_i]{
	    f:\bang[0]\lm(A\to B),\ u:\bang[0]\lm A
	    \vdash
	    \abs k[\ito[1]{B}{\bot}].
	    f(\abs g[A\to B].u(\abs x[A].k(gx))):
	    \lm B
	  }
	  \infer1[\lto_i]{
	    f:\bang[0]\lm(A\to B)
	    \vdash
	    \abs u[\lm A].
	    \abs k[\ito[1]{B}{\bot}].
	    f(\abs g[A\to B].u(\abs x[A].k(gx))):
	    \lm A\to\lm B
	  }
	  \infer1[\lto_i]{
	    \vdash\termstyle{gapp}_{A,B}:
	    \lm(A\to B)\to\lm A\to\lm B
	  }
	\end{prooftree}%
}%
}
\caption{Typings for $\termstyle{gfix}_A$ and $\termstyle{gapp}_{A,B}$.}
\label{fig:app:g:gfix-gapp}
\end{figure}

\clearpage

\subsection{Lemmas for the Productivity Proof}
\label[appendix]{app:g:ssec:productivity}

We prove \cref{lem:if-typable-then-in-rc,lem:forget-hn},
which are used in the proof of \cref{thm:typable-productive}.

\linkedrestate{lem:if-typable-then-in-rc}{\sysGiftypabletheninrc}

	\begin{proof}
	As customary we prove the more general
	\cref{lem:app:g:if-typable-then-in-rc-extended} below.
	Then our goal is the particular case
	where $T_i \eqdef x_i$ and $\rho \eqdef \rho_\HN$.
	\qed
	\end{proof}

\begin{lemma}
\label{lem:app:g:if-typable-then-in-rc-extended}
	If $x_1:\bang[a_1]C_1, \dots, x_n:\bang[a_n]C_n \vdash M:A$,
	$\rho$ is an environment, and
	\[	\text{for all $i \in [1,n]$,}\ 
		T_i \in \redinterp{C_i}[a_i] \eqdef \begin{cases}
		\redinterp{C_i} & \text{if $a_i = 0$,} \\
		\pureterms & \text{otherwise,}
		\end{cases} \]
	then $\forgetFmu M \subst[\vec x]{\vec T}
	\eqdef \forgetFmu M \subst[x_1]{T_1} \cdots \subst[x_n]{T_n}
	\in \redinterp A$.
\end{lemma}

	\begin{proof}
	We do so by induction on $M$, or equivalently on the derivation
	$x_1:\bang[a_1]C_1, \dots, x_n:\bang[a_n]C_n \vdash M:A$.
	We write $\Gamma$ for $x_1:\bang[a_1]C_1, \dots, x_n:\bang[a_n]C_n$.
	
	\begin{itemize}
	\item If the last rule of the derivation is \drule{ax},
		with conclusion $\Gamma \vdash x_i:C_i$
		for some $i \in [1,n]$ such that $a_i = 0$,
		then $\forgetFmu{x_i} \subst[\vec x]{\vec T}
		= T_i \in \redinterp{C_i}[0] = \redinterp{C_i}$.
	
	\item If the last rule of the derivation is \drule{\lto_i}
		with conclusion $\Gamma \vdash \abs x[A].M : \ito AB$,
		then we need to show that
		$\forgetFmu{\abs x[A].M} \subst[\vec x]{\vec T}
		= \abs x.\forgetFmu M \subst[\vec x]{\vec T}
		\in \redinterp{\ito AB}
		= \set{ U \in \pureterms }
			[ \forall V \in \redinterp{A}[a],\ UV \in \redinterp B ]$.
		Take $V \in \redinterp{A}[a]$, then
		$(\abs x.\forgetFmu M \subst[\vec x]{\vec T}) V
		\hred \forgetFmu M \subst[\vec x]{\vec T} \subst V
		\in \redinterp B$ by the induction hypothesis.
		We conclude by \cref{def:rc}, clause (3).
		
	\item The key case is when
		the last rule of the derivation is \drule{\lto_e}, \ie
		\[	\begin{prooftree}
			\hypo{ x_1:\bang[b_1]C_1, \dots, x_n:\bang[b_n]C_n
				\vdash M : \ito A B }
			\hypo{ x_1:\bang[c_1]C_1, \dots, x_n:\bang[c_n]C_n
				\vdash N:A }
			\infer2[\lto_e]{ \Gamma \vdash MN:B  }
			\end{prooftree} \]
		such that $\forall i \in [1,n],\ a_i = \min(b_i,c_i+a)$.
		Notice that some of the $b_i$ (\resp $c_i$)
		may be equal to $\infty$
		when the corresponding $x_i$ is not in the support of the context
		of the left hypothesis (\resp right one).
		Take an environment $\rho$ and, for each $i \in [1,n]$,
		a term $T_i \in \redinterp{C_i}[a_i]$.
		Observe that $a_i \leq b_i$ implies that
		$\redinterp{C_i}[a_i] \subseteq \redinterp{C_i}[b_i]$,
		giving rise to an induction hypothesis on the left sub-derivation:
		$\forgetFmu M \subst[\vec x]{\vec T}
		\in \redinterp{\ito AB}
		= \set{ U \in \pureterms }
			[ \forall V \in \redinterp{A}[a],\ UV \in \redinterp B ]$.
		\begin{itemize}
		\item If $a = 0$, similarly we have $a_i \leq c_i$,
			giving rise to an induction hypothesis 
			on the right sub-derivation:
			$\forgetFmu N \subst[\vec x]{\vec T} \in \redinterp A
			= \redinterp{A}[a]$.
		\item If $a > 0$, $\forgetFmu N \subst[\vec x]{\vec T}
			\in \pureterms = \redinterp{A}[a]$ is immediate.
		\end{itemize}
		In both cases we can conclude that
		$\forgetFmu{MN} \subst[\vec x]{\vec T}
		= \forgetFmu M \subst[\vec x]{\vec T}
			\forgetFmu N \subst[\vec x]{\vec T}
		\in \redinterp B$.
	
	\item If the last rule of the derivation is \drule{\fatype_i}
		with conclusion $\Gamma \vdash \typeabs α.M : \fatype α.B$,
		then we need to show that
		$\forgetFmu{\typeabs α.M} \subst[\vec x]{\vec T}
		= \forgetFmu M \subst[\vec x]{\vec T}
		\in \redinterp{\fatype α.B}
		= \bigcap_{\rcR} \redinterp[\rho \redsubst \rcR] B$.
		Take an environment $\rho$ and, for each $i \in [1,n]$,
		a term $T_i \in \redinterp{C_i}[a_i]$.
		Fix an \rc $\rcR$. For all $i \in [1,n]$:
		\begin{itemize}
		\item If $a_i = 0$, consider the following fact:
			\begin{quote}
				For all $T \in \pureterms$, $C \in \fmutypes$,
				environment $\rho$,	$\alpha \in \Atoms$ and \rc $\rcR$, \\
				if $T \in \redinterp C$ and $α \notin \fv(C)$
				then $T \in \redinterp[\rho \redsubst \rcR] C$.
			\end{quote}
			(It can be proved by an immediate induction over $C$.)
			Using the hypothesis $α \notin \fv(\Gamma) \supseteq \fv(C_i)$,
			we obtain $T_i \in \redinterp[\rho \redsubst \rcR]{C_i}
			= \redinterp[\rho \redsubst \rcR]{C_i}[a_i]$.
		\item If $a_i > 0$, then
			$T_i \in \redinterp{C_i}[a_i] = \pureterms
			= \redinterp[\rho \redsubst \rcR]{C_i}[a_i]$.
		\end{itemize}
		Hence we can apply the induction hypothesis to $\rho \redsubst \rcR$
		and $T_1,\dots,T_n$, obtaining
		$\forgetFmu M \subst[\vec x]{\vec T} 
		\in \redinterp[\rho \redsubst \rcR] B$.
	
	\item If the last rule of the derivation is \drule{\fatype_e}
		with conclusion $\Gamma \vdash MA : B \tsubst A$,
		consider the following fact:
		\begin{quote}
			For all $A,B \in \fmutypes$, $α \in \Atoms$ and environment 
			$\rho$,
			$\redinterp{ B \tsubst A }
			= \redinterp[\rho \redsubst{\redinterp A}] B$.
		\end{quote}
		(It can be proved by an immediate induction over $B$.)
		Take an environment $\rho$ and, for each $i \in [1,n]$,
		a term $T_i \in \redinterp{C_i}[a_i]$.
		By first the induction hypothesis, then the above fact,
		$\forgetFmu{MA} \subst[\vec x]{\vec T}
		= \forgetFmu M \subst[\vec x]{\vec T}
		\in \redinterp{\fatype α.B}
		= \bigcap_{\rcR} \redinterp[\rho \redsubst \rcR] B
		\subseteq \redinterp[\rho \redsubst{\redinterp A}] B
		= \redinterp{ B \tsubst A }.$
	
	\item If the last rule of the derivation is \drule{fold}
		with conclusion $\Gamma \vdash \fold M : \rectype α.A$,
		then for all environment $\rho$
		and $T_i \in \redinterp{C_i}[a_i]$ for $i \in [1,n]$,
		the induction hypothesis yields
		$\forgetFmu{\fold M} \subst[\vec x]{\vec T}
		= \forgetFmu M \subst[\vec x]{\vec T}
		\in \redinterp{A \tsubst{\rectype α.A}}
		= \redinterp{\rectype α.A}$.
	
	\item If the last rule of the derivation is \drule{unfold}
		with conclusion $\Gamma \vdash \unfold M : A \tsubst{\rectype α.A}$,
		then for all environment $\rho$
		and $T_i \in \redinterp{C_i}[a_i]$ for $i \in [1,n]$,
		the induction hypothesis yields
		$\forgetFmu{\unfold M} \subst[\vec x]{\vec T}
		= \forgetFmu M \subst[\vec x]{\vec T}
		\in \redinterp{\rectype α.A}
		= \redinterp{A \tsubst{\rectype α.A}}$.
	\qed
	\end{itemize}
	\end{proof}

\linkedrestate{lem:forget-hn}{\sysGforgethn}

	\begin{proof}
	Write $M \hred[\beta] N$, $M \hred[\fatype] N$ and $M \hred[\rectype] N$
	when a term $M$ head reduces to $N$ and the fired head redex
	corresponds to respectively the first (application of a λ-abstraction),
	second (application of $\typeabs$-abstraction),
	and third (unfolding of a fold) case from \cref{def:bred}.
	Observe the following facts:
	\begin{enumerate}
		\item If $M \hred[\fatype] N$ or $M \hred[\rectype] N$
		then $N$ contains two term constructors less than $M$.
		\item If $M \hred[\beta] N$ then $\forgetFmu M \hred \forgetFmu N$.
		If $M \hred[\fatype] N$ or $M \hred[\rectype] N$
		then $\forgetFmu M = \forgetFmu N$.
	\end{enumerate}
	Now we proceed by contraposition:
	suppose $M$ is not head normalising,
	then there is an infinite sequence of head reductions starting from $M$.
	By the first above fact, infinitely many of these head reductions
	are $\hred[\beta]$ reductions.
	By the second above fact, there is an infinite sequence of head
	reductions starting from $\forgetFmu M$.
	By a standard result \autocite[Theorem~8.3.11]{Barendregt},
	$\forgetFmu M$ is not head normalising.
	\qed
	\end{proof}

\clearpage
\section[Appendix: Proofs about I\_T]{Proofs about $\itypes$ (\cref{sec:it})}
\label[appendix]{app:it}
% !TeX root = ../main.tex
% !TeX spellcheck = en_GB

\subsection{Verifying the Definition of the Ultrametric Distances}
\label[appendix]{app:i:ssec:ultrametric}

\linkedrestate{lem:dist-is-ultrametric}{\sysIdistisultrametric}

To prove this, let us introduce the following definition
and characterization.

\begin{definition}
	The \emph{truncation at depth $d$} of a type or a multi-type
	is defined for all $d \in \Tropical*$ by induction by:
	\begin{gather*}
		\trunc\alpha \eqdef \alpha \qquad
		\trunc{\ms m \lto \sigma} \eqdef \trunc{\ms m} \lto \trunc\sigma \\
		\trunc{\ms m} : \tau \mapsto \left( a \mapsto \begin{cases}
			\displaystyle\sum_{\substack{ 
				\sigma \in \itypes[\msets{\Tropical*}] \\ 
				\trunc[d-a]\sigma = \tau
			}} \ms m(\sigma)(a) & \text{if $a \leq d$,} \\
			0 & \text{otherwise}
		\end{cases} \right).
	\end{gather*}
	
	The definition is by induction because each recursive call
	either corresponds to an inductive formation rule ($a = 0$),
	or makes the depth $d$ decrease ($1 \leq a \leq d$):
	the definition says that
	if $\sigma$ bears grade $a \leq d$ in $\ms m$
	(rigorously, if $a$ is among the multiset of grades borne by $\sigma$
	in $\ms m$),
	then $\trunc[d-a]\sigma$ bears grade $a$ in $\trunc{\ms m}$;
	all grades above $d$ in $\ms m$ are forgotten.
\end{definition}

\begin{lemma} \label{lem:dist-iff-trunc}
	For all $\sigma, \tau \in \itypes[\msets{\Tropical*}]$
	and $d \in \Tropical*$,
	\[	\dist(\sigma, \tau) < 2 ^{-d}
		\quad\text{\ifandonlyif}\quad
		\trunc\sigma = \trunc\tau, \]
	and similarly for all $\ms m,\ms n \in \itypes![\msets{\Tropical*}]$.
\end{lemma}
	
	\begin{proof}
	The only difficult case is the one of
	$\ms m,\ms n \in \itypes![\msets{\Tropical*}]$:
	\begin{align*}
		&&& \dist(\ms m, \ms n) < 2^{-d} \\
		\iff &&& \exists \pi \in \pairings{\ms m}{\ms n},\ 
			\forall ((\sigma,a),(\tau,b)) \in \totpairing{\pi},\ 
			\cost((\sigma,a),(\tau,b)) < 2^{-d} \\
		\iff &&& \exists \pi \in \pairings{\ms m}{\ms n},\ 
			\begin{cases}
			\forall ((\sigma,a),(\tau,b)) \in \pi,\ 
				\begin{cases}
				\dist(a,b) < 2^{-d} & \text{if $a \neq b$} \\
				2^{-a} \dist(\sigma,\tau) < 2^{-d} & \text{otherwise}
				\end{cases} \\
			\forall (\sigma,a) \in \deletions{\pi},\ 2^{-a} < 2^{-d} \\
			\forall (\tau,b) \in \insertions{\pi},\ 2^{-b} < 2^{-d}
			\end{cases} \\
		\iff &&& \exists \pi \in \pairings{\ms m}{\ms n},\ 
			\begin{cases}
				\forall ((\sigma,a),(\tau,b)) \in \pi,\ 
					\begin{aligned}[t]
					& a,b > d \text{ or } \\
					& a = b \leq d \text{ and }
						\dist(\sigma,\tau) < 2^{-(d-a)} \leq 1
					\end{aligned} \\
				\forall (\sigma,a) \in \deletions{\pi},\ a > d \\
				\forall (\tau,b) \in \insertions{\pi},\ b > d
			\end{cases} \\
		\iff &&& \exists \pi \in \pairings{\ms m}{\ms n},\ 
			\forall ((\sigma,a),(\tau,b)) \in \totpairing{\pi},\ 
			\text{if $a \leq d$ or $b \leq d$ then $a=b$ and
			$\trunc[d-a]\sigma = \trunc[d-b]\tau$} \\
		\iff &&& \trunc {\ms m} = \trunc {\ms n}.
	\end{align*}
	\qed
	\end{proof}
	
	\begin{proof}[Proof of \cref{lem:dist-is-ultrametric}]
	First we prove that $\dist$ actually equips $\itypes[\msets{\Tropical*}]$
	and $\itypes![\msets{\Tropical*}]$ with an ultrametric structure.
	\begin{itemize}
	\item The fact that $\dist(\sigma,\tau) = 0 \iff \sigma = \tau$
		is an direct consequence of \cref{lem:dist-iff-trunc}.
		Same for multi-types.
	\item The fact that $\dist(\sigma,\tau) = \dist(\tau, \sigma)$,
		and same for multi-types, is immediate.
	\item The ultrametric triangle inequality is the difficult part.
		In particular the crucial case is again the one of multi-types.
		Let us compute:
		\begin{align*}
		&&& \max(\dist(\ms m, \ms n), \dist(\ms n, \ms o)) \\
		= &&& \max\left( \inf\limits_{\pi \in \pairings{\ms m}{\ms n}}
				\max\limits_{p \in \totpairing{\pi}} \cost(p),\ 
				\inf\limits_{\rho \in \pairings{\ms n}{\ms o}}
				\max\limits_{q \in \totpairing{\rho}} \cost(q)
			\right) \\
		= &&& \inf\limits_{\pi, \rho} \max\limits_{p,q}
			\max(\cost(p), \cost(q)) \\
		\geq &&& \inf\limits_{\pi, \rho} \max\left(
			\begin{aligned}
			& \max\limits_{\substack{ ((\sigma,a),(\tau,b)) \in \pi \\
				((\tau,b),(\upsilon,c)) \in \rho }}
				\max\left( \cost((\sigma,a),(\tau,b)),
					\cost((\tau,b),(\upsilon,c)) \right), \\
			& \max\limits_{\substack{(\sigma,a) \in \deletions{\pi}}}
				\cost((\sigma,a), \bullet),
			\max\limits_{\substack{ ((\sigma,a),(\tau,b)) \in \pi \\
				(\tau,b) \in \deletions{\rho} }}
				\max\left( \cost((\sigma,a),(\tau,b)),
					\cost((\tau,b),\bullet) \right), \\
			& \max\limits_{\substack{ (\tau,b) \in \insertions{\pi} \\
				((\tau,b),(\upsilon,c)) \in \rho }}
				\max\left( \cost(\bullet, (\tau,b)),
					\cost((\tau,b),(\upsilon,c)) \right),
			\max\limits_{\substack{(\upsilon,c) \in \insertions{\rho}}}
				\cost(\bullet, (\upsilon,c))
			\end{aligned} \right) \\
		\intertext{where the first line corresponds to the matchings of the
		composition $\rho \circ \pi$, and the second and thid lines
		respectively to its deletions and insertions;
		elementary computations in each maximum,
		using the ultrametric triangle inequality for
		$\dist$ on $\Tropical$ (fact) and
		$\dist$ on $\itypes[\msets\Tropical*]$ ((co)induction hypothesis),
		allow us to deduce:}
		\geq &&& \inf\limits_{\pi, \rho}
			\max\limits_{((\sigma,a),(\upsilon,c)) \in 
				\totpairing{\rho \circ \pi}}
			\cost((\sigma,a),(\upsilon,c)) \\
		\geq &&& \dist(\ms m, \ms o).
		\end{align*}
	\end{itemize}
	
	Completeness of the ultrametric structure is straightforward:
	if $(\sigma_n)_{n \in \Nat}$ is a Cauchy sequence,
	then there is an $N$ such that for all $n \geq N$,
	$\dist(\sigma_n, \sigma_N) < \frac 12$;
	this allows to build a formation derivation for a limit $\sigma$,
	working by induction in the inductive hypotheses of the formation rules,
	and by coinduction in the coinductive hypotheses that exactly correspond
	to recursive steps where the hypothesis
	$\dist(\sigma_n, \sigma_N) < \frac 12$ is not preserved.
	
	The fact that every type $\sigma$ is the limit of a Cauchy sequence of
	finite types is a consequence of \cref{lem:dist-iff-trunc}:
	it suffices to take the sequence $(\trunc[n]\sigma)_{n\in\Nat}$.
	\qed 
	\end{proof}

The following Lemma will also be used in the following.

\begin{lemma}
\label{lem:i:deriv-to-type-context-non-expansive}
	For all derivations $\derivD \derives \Gamma \vdash M:\sigma$
	and $\derivD' \derives \Gamma' \vdash M:\sigma'$
	in $\itypes[\msets{\Tropical*}]$,
	\[	\forall x\in\Var,\ \dist(\Gamma(x),\Gamma'(x)) 
			\leq \dist(\derivD,\derivD') \qquad
		\dist(\sigma,\sigma') \leq \dist(\derivD,\derivD'). \]
	The same holds for derivations in $\itypes![\msets{\Tropical*}]$.
\end{lemma}
	
	\begin{proof}
	Straightforward induction on $\derivD$ and $\derivD'$.
	\qed
	\end{proof}

\subsection{Non-Expansive Subject Expansion for $\itypes[\msets{\Tropical*}]$}
\label[appendix]{app:i:subject-expansion}

\linkedrestate{the:itset-itypes-subjexp}{\sysIsubjectexpansion}

We first prove a backward substitution lemma.
By relevance, every derivable judgment
$\Gamma\vdash P:\sigma$ satisfies
$\Gamma(z)=[]$ whenever $z\notin\fv(P)$;
the same holds for multi-type judgments.
We work up to $\alpha$-conversion, choosing bound variables
fresh for the substitutions under consideration.

\begin{lemma}[Backward Substitution Lemma]
\label{lem:i:backward-substitution}
	Let $x\notin\fv(N)$.
	For all derivation $\derivD\derives\Gamma\vdash M\subst N:\sigma$,
	one can build derivations
	$\derivD_{M,N,x}^l\derives\Gamma_M,x:\ms m\vdash M:\sigma$
	and
	$\derivD_{M,N,x}^r\derives\Gamma_N\vdash N:\ms m$
	such that
	$\Gamma=\Gamma_M\mstplus\Gamma_N$.
	The same holds for multi-type judgments.
	For fixed $M,N,x$ and all input derivations
	$\derivD,\derivE$,
	\[
		\dist(\derivD_{M,N,x}^l,\derivE_{M,N,x}^l)
			\leq\dist(\derivD,\derivE),
		\qquad
		\dist(\derivD_{M,N,x}^r,\derivE_{M,N,x}^r)
			\leq\dist(\derivD,\derivE).
	\]
	These constructions preserve finiteness
	and are effective on finite derivations.
\end{lemma}

\begin{proof}
	The construction of the derivations is standard;
	we emphasize how non-expansive\-ness is obtained.
	The proof is by induction on $M$,
	first for ordinary type judgments,
	then for multi-type judgments.
	We identify a multi-type derivation with the weighted
	multiset of premises of its final \drule{!} rule.

	\begin{itemize}
	\item If $M=x$, define
		\[
			\derivD_{x,N,x}^l
			\quad\eqdef\quad
			\begin{smallprooftree}[center]
				\infer0[ax]{
					x:[\gr{\sigma}{[0]}]\vdash x:\sigma
				}
			\end{smallprooftree}
			\qquad
			\derivD_{x,N,x}^r
			\quad\eqdef\quad
			\begin{smallprooftree}[center]
				\hypo{[\gr{\derivD}{[0]}]}
				\infer1[!]{
					\Gamma\vdash N:[\gr{\sigma}{[0]}]
				}
			\end{smallprooftree}.
		\]
		For $\derivE\derives\Delta\vdash N:\tau$,
		we have
		$\dist(\derivD_{x,N,x}^l,\derivE_{x,N,x}^l)
		=\dist(\sigma,\tau)\leq\dist(\derivD,\derivE)$
		by \cref{lem:i:deriv-to-type-context-non-expansive},
		and
		$\dist(\derivD_{x,N,x}^r,\derivE_{x,N,x}^r)
		=\dist(\derivD,\derivE)$.

	\item If $M=y\neq x$, take
		$\derivD_{y,N,x}^l\eqdef\derivD$
		and let $\derivD_{y,N,x}^r$ be the nullary
		\drule{!} derivation of $\vdash N:[]$.
		Both estimates are immediate.

	\item If $M=\lambda y.P$, choose $y\notin\fv(N)\cup\set x$.
		Apply the induction hypothesis to the premise of $\derivD$.
		By relevance, $\Gamma_N(y)=[]$.
		Abstract $y$ in the left output and retain the right output.
		Both estimates follow from the induction hypothesis.

	\item If $M=PQ$, write
		\[
			\derivD
			\quad=\quad
			\begin{smallprooftree}[center]
				\hypo{
					\derivD'\derives\Gamma'
					\vdash P\subst N:\ms n\lto\sigma
				}
				\hypo{
					\derivD''\derives\Gamma''
					\vdash Q\subst N:\ms n
				}
				\infer2[\lto_e]{
					\Gamma'\mstplus\Gamma''
					\vdash(PQ)\subst N:\sigma
				}
			\end{smallprooftree}
		\]
		and form
		\begin{gather*}
			\derivD_{PQ,N,x}^l
			\quad\eqdef\quad
			\begin{smallprooftree}[center]
				\hypo{
					(\derivD')_{P,N,x}^l
					\derives\Gamma'_P,x:\ms m_P
					\vdash P:\ms n\lto\sigma
				}
				\hypo{
					(\derivD'')_{Q,N,x}^l
					\derives\Gamma''_Q,x:\ms m_Q
					\vdash Q:\ms n
				}
				\infer2[\lto_e]{
					\Gamma'_P\mstplus\Gamma''_Q,
					x:(\ms m_P\mstplus\ms m_Q)
					\vdash PQ:\sigma
				}
			\end{smallprooftree},
			\\[\topsep]
			\derivD_{PQ,N,x}^r
			\quad\eqdef\quad
			\begin{smallprooftree}[center]
				\hypo{
					(\derivD')_{P,N,x}^r
					\mstplus
					(\derivD'')_{Q,N,x}^r
				}
				\infer1[!]{
					\Gamma'_N\mstplus\Gamma''_N
					\vdash N:\ms m_P\mstplus\ms m_Q
				}
			\end{smallprooftree}.
		\end{gather*}
		Unions of pairings give
		$\dist(\ms U\mstplus\ms V,\ms U'\mstplus\ms V')
		\leq\max(\dist(\ms U,\ms U'),\dist(\ms V,\ms V'))$.
		Thus
		\begin{align*}
			\dist(\derivD_{PQ,N,x}^l,\derivE_{PQ,N,x}^l)
			&=\max\bigl(
				\dist((\derivD')_{P,N,x}^l,
					(\derivE')_{P,N,x}^l),
				\dist((\derivD'')_{Q,N,x}^l,
					(\derivE'')_{Q,N,x}^l)
			\bigr)
			\\
			&\leq\max(\dist(\derivD',\derivE'),
				\dist(\derivD'',\derivE'')) \\
			&=\dist(\derivD,\derivE),
			\\
			\dist(\derivD_{PQ,N,x}^r,\derivE_{PQ,N,x}^r)
			&\leq\max\bigl(
				\dist((\derivD')_{P,N,x}^r,
					(\derivE')_{P,N,x}^r),
				\dist((\derivD'')_{Q,N,x}^r,
					(\derivE'')_{Q,N,x}^r)
			\bigr)
			\\
			&\leq\max(\dist(\derivD',\derivE'),
				\dist(\derivD'',\derivE'')) \\
			&=\dist(\derivD,\derivE).
		\end{align*}
	\end{itemize}

	For a multi-type judgment, write the premises of the final
	\drule{!} rule as
	$[\gr{\derivD_i}{\ms a_i}]_{i\in I}$.
	Define
	\[
		\derivD_{M,N,x}^l
			\eqdef
			[\gr{(\derivD_i)_{M,N,x}^l}{\ms a_i}]_{i\in I},
		\qquad
		\derivD_{M,N,x}^r
			\eqdef
			\mstsum_{i\in I}
				\ms a_i\msttimes(\derivD_i)_{M,N,x}^r.
	\]
	The extracted multi-type and contexts are
	$\ms m=\mstsum_i\ms a_i\msttimes\ms m_i$,
	$\Gamma_M=\mstsum_i\ms a_i\msttimes\Gamma_{i,M}$
	and
	$\Gamma_N=\mstsum_i\ms a_i\msttimes\Gamma_{i,N}$.
	The context equation follows by distributivity.

	For weighted multisets of derivations $\ms U,\ms V$
	and grades $a,b\in\Tropical*$,
	\begin{gather*}
		\dist([a]\msttimes\ms U,[b]\msttimes\ms V)
		\leq
		\begin{cases}
			2^{-a}\dist(\ms U,\ms V) & \text{if $a=b$,}\\
			2^{-\min(a,b)} & \text{otherwise,}
		\end{cases}
		\\
		\dist([a]\msttimes\ms U,[])\leq2^{-a}.
	\end{gather*}
	Indeed, shift a pairing by $a$ when $a=b$;
	otherwise delete and insert all occurrences.
	Consequently, every match
	$((\derivD_i,a),(\derivE_j,b))$
	can be replaced by pairings of the right outputs
	with infimal cost at most
	$\cost((\derivD_i,a),(\derivE_j,b))$;
	the same holds for deletions and insertions.
	For the left outputs, apply the induction hypothesis
	directly to matched premises.
	Taking unions of the induced pairings and then infima gives,
	for $u\in\{l,r\}$,
	\[
		\dist(\derivD_{M,N,x}^u,\derivE_{M,N,x}^u)
		\leq
		\inf_{\pi\in\pairings{\derivD}{\derivE}}
			\max_{p\in\totpairing\pi}\cost(p)
		=\dist(\derivD,\derivE).
	\]

	Preservation of finiteness and effectiveness
	follow from the construction.
	\qed
\end{proof}

\begin{proof}[Proof of \cref{the:itset-itypes-subjexp}]
	For a root contraction
	$(\lambda x.P)N\bred P\subst N$,
	choose $x\notin\fv(N)$ and define
	\[
		\subjexp{(\lambda x.P)N\bred P\subst N}(\derivD)
		\quad\eqdef\quad
		\begin{smallprooftree}[center]
			\hypo{
				\derivD_{P,N,x}^l
				\derives\Gamma_P,x:\ms m\vdash P:\sigma
			}
			\infer1[\lto_i]{
				\Gamma_P\vdash\lambda x.P:\ms m\lto\sigma
			}
			\hypo{
				\derivD_{P,N,x}^r
				\derives\Gamma_N\vdash N:\ms m
			}
			\infer2[\lto_e]{
				\Gamma_P\mstplus\Gamma_N
				\vdash(\lambda x.P)N:\sigma
			}
		\end{smallprooftree}.
	\]
	By \cref{lem:i:backward-substitution},
	the conclusion context and type are preserved, and
	\begin{multline*}
		\dist\bigl(
			\subjexp{(\lambda x.P)N\bred P\subst N}(\derivD),
			\subjexp{(\lambda x.P)N\bred P\subst N}(\derivE)
		\bigr)
		\\ =
		\max_{u\in\{l,r\}}
			\dist(\derivD_{P,N,x}^u,\derivE_{P,N,x}^u)
		\leq\dist(\derivD,\derivE).
	\end{multline*}
	For a contraction inside a term, extend the construction
	through $\drule{\lto_i}$ and $\drule{\lto_e}$,
	and pointwise through \drule{!}.
	These extensions preserve non-expansiveness and finiteness.
	For a finite reduction sequence, compose the resulting maps;
	for the empty sequence, take the identity.
	\qed
\end{proof}

\subsection{Positive and Negative Occurrences of the Empty $\msets{\Tropical*}$-Set.}
\label[appendix]{app:i:positivity}

This section contains technical material that will be used in the next Section, \ie, in the proof of \cref{lem:it:guarded-approximation}.

The positive and negative subsets $\itypes[\msets{\Tropical*}]+$, etc., are not closed in the surrounding ultrametric spaces, so that an additional control on convergence is needed. This is the condition of \emph{persistent positivity} (\cref{def:it:persistent-positivity}) ensuring
that positive occurrences of $[]$ do not appear at the limit.
To do so,
we first introduce \emph{graded paths} as a way do identify multiset occurrences inside types and derivations.

\begin{definition}[Graded paths]
\label{def:it:paths}
	A \emph{graded path} is a finite sequence of transitions of the following form:
	\begin{align*}
	(\ms m\lto\sigma,\epsilon,d)
		& \leadsto(\ms m,\bar\epsilon,d), \\
	(\ms m\lto\sigma,\epsilon,d)
		& \leadsto(\sigma,\epsilon,d), \\
	(\ms m,\epsilon,d)
		& \leadsto(\sigma,\epsilon,d+a)
		\quad\text{if $\ms m(\sigma)(a) \neq 0$.}
	\end{align*}
	Here $\epsilon\in\{+,-\}$,
	$\bar\epsilon$ denotes the opposite polarity,
	and $a,d\in\Tropical*$.
	The sequence may contain no transition.
	The \emph{polarity} and \emph{temporal depth} of a path
	are the second and third components of its final state.
	We denote its temporal depth by $\temporaldepth{p}$.
	
	The set $\gradingpaths\epsilon(\sigma)$ of \emph{graded paths in $\sigma$ with initial polarity $\epsilon$} is the set of all graded paths starting at $(\sigma, \epsilon, 0)$.
	We define $\gradingpaths{\epsilon}(\ms m)$ similarly.

	A \emph{graded path in a judgment} $\Gamma\vdash M:\sigma$
	is either a pair $(\bullet, p)$ for some $p \in \gradingpaths+(\sigma)$,
	or a pair $(x,p)$ for some $x\in\Var$ and $p \in \gradingpaths-(\Gamma(x))$.
	A \emph{graded path in a derivation} is a graded path in its conclusion judgement.
\end{definition}

The following definition stratifies \cref{def:it:aux-positive}.

\begin{definition}[Stratified positivity]
\label{def:it:approx-positive}
	For every $n \in \Tropical*$, define the following subsets
	of types and multi-types by:
	\begin{align*}
		\sigma\in\itypes[\msets{\Tropical*}]+[n]
		&\quad\text{whenever}\quad
		\forall d<n,\
		(\sigma,+,0)\not\leadsto^*([],+,d),
		\\
		\sigma\in\itypes[\msets{\Tropical*}]-[n]
		&\quad\text{whenever}\quad
		\forall d<n,\
		(\sigma,+,0)\not\leadsto^*([],-,d),
		\\
		\ms m\in\itypes![\msets{\Tropical*}]+[n]
		&\quad\text{whenever}\quad
		\forall d<n,\
		(\ms m,+,0)\not\leadsto^+([],+,d),
		\\
		\ms m\in\itypes![\msets{\Tropical*}]-[n]
		&\quad\text{whenever}\quad
		\forall d<n,\
		(\ms m,+,0)\not\leadsto^+([],-,d).
	\end{align*}
	Notice that
	$
		\itypes[\msets{\Tropical*}]+[0]
		=\itypes[\msets{\Tropical*}]-[0]
		=\itypes[\msets{\Tropical*}]
	$,
	and similarly for multi-types.
	We also define
	\[
		\itypes[\msets{\Tropical*}]+
		=\bigcap_{n\in\Nat}
			\itypes[\msets{\Tropical*}]+[n],
		\qquad
		\itypes[\msets{\Tropical*}]-
		=\bigcap_{n\in\Nat}
			\itypes[\msets{\Tropical*}]-[n],
	\]
	and similarly for multi-types
	(retrieving \cref{def:it:aux-positive}).
	We write
	$\derivD\derives+[n]\Gamma\vdash M:\sigma$
	whenever $\derivD\derives\Gamma\vdash M:\sigma$ with
	$\sigma\in\itypes[\msets{\Tropical*}]+[n]$, and
	$\Gamma : \Var \to \itypes![\msets{\Tropical*}]-[n]$.
\end{definition}

Alas, positivity is not preserved by convergence, even for typing derivations.
For example, put
$\tau_n\eqdef[\gr{\alpha}{[n]}]\lto\alpha$
and $\tau\eqdef[]\lto\alpha$.
The positive derivations
$\derivD_n \derives+ x:[\gr{\tau_n}{[0]}]\vdash x:\tau_n$
consisting of just an \drule{ax} rule
converge to a derivation of
$x:[\gr{\tau}{[0]}]\vdash x:\tau$,
which is not positive.
To control the occurrences of $[]$ in the limit of such a sequence of derivations,
we first explain how its graded paths can be lifted
to the elements of the sequence.

\begin{lemma}
\label{lem:it:path-lifting}
	Let $(\derivD_n)_{n\in\Nat}$ be a sequence of derivations
	$\derivD_n\derives\Gamma_n\vdash M:\sigma_n$
	converging to
	$\derivD\derives\Gamma\vdash M:\sigma$.
	For every graded path
	$
		p=(r,((T_i,\epsilon_i,d_i))_{0\leq i\leq k})
	$
	in $\derivD$,
	there are $N\in\Nat$ and, for $n\geq N$, graded paths
	$
		p_n=(r,((T_{n,i},\epsilon_i,d_i))_{0\leq i\leq k})
	$
	in $\derivD_n$, making the same choices of transitions,
	such that for all $i\in[0,k]$, $\lim_{n\to\infty}T_{n,i}=T_i$.
	We call such a sequence $(p_n)_{n\geq N}$
	a sequence of \emph{lifts} of $p$.
\end{lemma}

	\begin{proof}
	By induction on $k$.
	If $k=0$, the result follows from
	\cref{lem:i:deriv-to-type-context-non-expansive}.
	Suppose $k>0$.
	Apply the induction hypothesis to the prefix of $p$
	ending at $(T_{k-1},\epsilon_{k-1},d_{k-1})$.
	This gives lifts of that prefix for all $n\geq N$,
	whose final first components $T_{n,k-1}$
	converge to $T_{k-1}$.
	We extend these lifts according to the last transition of $p$.
	
	\begin{itemize}
	\item
	If $T_{k-1}=\ms m\lto\tau$, then, after increasing $N$ if needed,
	$
		T_{n,k-1}=\ms m_n\lto\tau_n,
	$
	since an atom and an arrow are at distance $1$.
	Extend each lifted prefix into the domain or the codomain,
	according to the transition chosen by $p$.
	The resulting sequence $T_{n,k}$ converges to $T_k$ by \cref{def:dist-on-itypes}.
	The resulting final states have the required polarity
	and depth.
	
	\item
	Otherwise, the last transition is
	$(\ms m,\epsilon,d) \leadsto(\tau,\epsilon,d+a)$
	with $\ms m(\tau)(a)\neq0$.
	Write $T_{n,k-1}=\ms m_n$, so $\ms m_n\to\ms m$.
	Set
	$
		\eta_n\eqdef\dist(\ms m,\ms m_n)+2^{-n}.
	$
	By \cref{def:bottleneck-dist}, choose a pairing
	between $\ms m$ and $\ms m_n$
	whose total cost is less than $\eta_n$.
	Increase $N$ so that $\eta_n<2^{-a}$ for every $n\geq N$.

	No occurrence of $(\tau,a)$ can be deleted by this pairing,
	since its deletion cost would be $2^{-a}$.
	Choose an occurrence matched with $(\tau,a)$,
	and denote it by $(\tau_n,b_n)$.
	If $b_n\neq a$, its matching cost would satisfy
	\[
		\cost((\tau,a),(\tau_n,b_n))
		\geq\dist(a,b_n)
		=2^{-\min(a,b_n)}
		\geq2^{-a},
	\]
	a contradiction.
	Thus $b_n=a$ and
	\[
		\ms m_n(\tau_n)(a)\neq0,
		\qquad
		2^{-a}\dist(\tau,\tau_n)
		=\cost((\tau,a),(\tau_n,a))
		<\eta_n.
	\]
	Consequently $\tau_n\to\tau$, and we extend the lifted prefix by
	$(\ms m_n,\epsilon,d) \leadsto(\tau_n,\epsilon,d+a)$.
	This preserves the selected grade, hence the temporal depth,
	and completes the induction.
	\qed
	\end{itemize}
	\end{proof}

The lemma ensures that each multi-type element selected by a graded path
of the limit $\derivD$ is found at the same grade in sufficiently late
approximations $\derivD_n$.
However, a multi-type in $\derivD$ a path may be empty
if all the corresponding elements in the derivations $\derivD_n$ move to arbitrarily large grades when $n$ increases,
as already underlined.
To exclude this behaviour, we introduce \emph{persistent positivity}.
Its definition relies on the following two notations.
For any multi-type $\ms m$, define its minimal grade by
$
	\mingrade(\ms m)
	\eqdef
	\min\set{a\in\Tropical*}
		[\exists\sigma,\ \ms m(\sigma)(a)\neq0]
$.
By \cref{def:bottleneck-dist},
\begin{equation}
\label{eq:it:distance-to-empty}
	\dist(\ms m,[])=2^{-\mingrade(\ms m)}.
\end{equation}
For a graded path $p$ ending at a multi-type,
we denote this multi-type by $\msofpath(p)$.

\begin{definition}[Persistent positivity]
\label{def:it:persistent-positivity}
	A convergent sequence $(\derivD_n)_{n\in\Nat}$
	with limit $\derivD$ is \emph{persistently positive}
	whenever, for every graded path $p$ in $\derivD$
	with positive polarity and ending at a multi-type,
	there are $b,N\in\Tropical*$
	and lifts $(p_n)_{n\geq N}$ as in
	\cref{lem:it:path-lifting} such that
	\[
		\forall n\geq N,\ \mingrade(\msofpath(p_n))\leq b.
	\]
\end{definition}

\begin{lemma}[Persistently positive convergence]
\label{lem:approx-positive}
	Consider a convergent sequence of derivations
	$\derivD_n\derives\Gamma_n\vdash M:\sigma_n$.
	Then it is persistently positive \ifandonlyif
	$\lim_n\derivD_n \derives+ \lim_n\Gamma_n \vdash M : \lim_n\sigma_n$.
\end{lemma}

	\begin{proof}
	Suppose the sequence is persistently positive.
	If $\lim_n\derivD_n$ is not positive, then it has a graded path $p$
	with positive polarity and $\msofpath(p)=[]$.
	With the notations of \cref{def:it:persistent-positivity},
	$\lim_n \msofpath(p_n) = []$, whereas for all $n \geq N$,
	\[
		\dist(\msofpath(p_n),[])
		=2^{-\mingrade(\msofpath(p_n))}
		\geq2^{-b}>0,
	\]
	which is a contradiction.

	Conversely, suppose $\lim_n\derivD_n$ is positive.
	Let $p$ be a graded path with positive polarity
	ending at a multi-type.
	Then $\msofpath(p)\neq[]$, so
	$b\eqdef\mingrade(\msofpath(p))$ is finite.
	Choose lifts by \cref{lem:it:path-lifting}.
	For sufficiently large $n$,
	\[
		\dist(\msofpath(p_n),\msofpath(p))<2^{-b}.
	\]
	By the ultrametric inequality and
	\cref{eq:it:distance-to-empty},
	\[
		\dist(\msofpath(p_n),[])
		=\dist(\msofpath(p),[])
		=2^{-b},
	\]
	hence $\mingrade(\msofpath(p_n))=b$.
	\qed
	\end{proof}

\subsection{Guarded Approximation}
\label[appendix]{app:i:guarded-approx}

\linkedrestate{lem:it:guarded-approximation}{\sysIguardedapprox}

\begin{proof}
	Finiteness follows by induction on $n$,
	since each $H_M$ has finitely many arguments
	and each $F_M$ preserves finiteness.

	By non-expansiveness of $F_M$,
	\cref{lem:i:deriv-to-type-context-non-expansive},
	and the definitions of the distances,
	for all $p,q\in\Nat$,
	\begin{align*}
		\dist(\derivD_{p+1}(M),\derivD_{q+1}(M))
		&\leq
			\dist(\derivD'_{p+1}(M),\derivD'_{q+1}(M))
		\\
		&\leq
			\frac12\max_{1\leq j\leq l}
				\dist(\derivD_p(M_j),\derivD_q(M_j)),
	\end{align*}
	with $\max\varnothing\eqdef0$.
	Indeed, the argument derivations are promoted at grade $[1]$,
	and the head types differ only inside the corresponding
	grade-$[1]$ domains.
	Induction on $\min(p,q)$, simultaneously for all
	$M\in\mathcal C$, therefore gives
	\[
		\dist(\derivD_p(M),\derivD_q(M))
			\leq2^{-\min(p,q)}.
	\]
	By \cref{lem:dist-on-derivations},
	the sequence converges to a derivation $\derivD(M)$.
	By \cref{lem:i:deriv-to-type-context-non-expansive},
	its conclusion is
	$\lim_n\Gamma_n(M)\vdash M:\lim_n\sigma_n(M)$.

	To prove positivity of this limit, we establish
	the following invariant:
	for every $n\in\Nat$, every $M\in\mathcal C$,
	and every graded path $p$ in $\derivD_n(M)$
	with positive polarity and ending at a multi-type,
	\begin{equation}
	\label{eq:it:guarded-approx-invariant}
		\temporaldepth{p}<n
		\quad\Longrightarrow\quad
		\mingrade(\msofpath(p))\leq1.
	\end{equation}
	The proof is by induction on $n$.
	The case $n=0$ is vacuous.
	Suppose the invariant holds at stage $n$.
	Fix $M\in\mathcal C$ and put
	\begin{align*}
		\theta
		&\eqdef
			[\gr{\sigma_n(M_1)}{[1]}]\lto\dots\lto
			[\gr{\sigma_n(M_l)}{[1]}]\lto\alpha,
		\\
		\Delta
		&\eqdef
			y:[\gr{\theta}{[0]}]
			\mstplus
			\mstsum_{j=1}^l[1]\msttimes\Gamma_n(M_j).
	\end{align*}
	Before the abstractions, $\derivD'_{n+1}(M)$ concludes
	$\Delta\vdash yM_1\dots M_l:\alpha$.
	Abstraction preserves the polarities and temporal depths
	of paths starting in context entries.
	Since $F_M$ preserves the conclusion judgment,
	it suffices to prove the invariant for this judgment.

	Let $p$ be a graded path in this judgment
	with positive polarity and ending at a multi-type.
	Since its result type is $\alpha$,
	$p$ starts in some $\Delta(z)$ with negative polarity.
	The first selected occurrence comes either from
	$y:[\gr{\theta}{[0]}]$
	or from some $[1]\msttimes\Gamma_n(M_j)$.

	\begin{itemize}
	\item
	If $p$ selects $\theta$ in the head contribution,
	it either ends at a domain
	$[\gr{\sigma_n(M_j)}{[1]}]$,
	whose minimal grade is $1$,
	or enters $\sigma_n(M_j)$ at grade $1$.
	In the latter case, its suffix, with all temporal depths
	decreased by $1$, defines a graded path $q$
	in $\derivD_n(M_j)$ with positive polarity.

	\item
	If $p$ selects an occurrence from
	$[1]\msttimes\Gamma_n(M_j)$,
	the corresponding occurrence in $\Gamma_n(M_j)$
	has grade one less.
	Selecting this occurrence and following the remaining
	transitions defines a graded path $q$
	in $\derivD_n(M_j)$ with positive polarity.
	\end{itemize}

	In both cases where $q$ is defined,
	\[
		\msofpath(q)=\msofpath(p),
		\qquad
		\temporaldepth{q}=\temporaldepth{p}-1.
	\]
	If $\temporaldepth{p}<n+1$,
	then $\temporaldepth{q}<n$,
	so the induction hypothesis gives
	$\mingrade(\msofpath(p))\leq1$.
	This proves \cref{eq:it:guarded-approx-invariant}.

	Finally, let $p$ be a graded path in $\derivD(M)$
	with positive polarity and ending at a multi-type.
	By \cref{lem:it:path-lifting},
	there are $N\in\Nat$ and lifts $(p_n)_{n\geq N}$
	with $\temporaldepth{p_n}=\temporaldepth{p}$.
	Increasing $N$ so that $N>\temporaldepth{p}$,
	\cref{eq:it:guarded-approx-invariant} gives
	\[
		\forall n\geq N,\quad
		\mingrade(\msofpath(p_n))\leq1.
	\]
	Thus $(\derivD_n(M))_{n\in\Nat}$
	is persistently positive.
	By \cref{lem:approx-positive},
	\[
		\derivD(M)
		\derives+
		\lim_n\Gamma_n(M)\vdash M:\lim_n\sigma_n(M).
	\]
	\qed
\end{proof}

\end{document}